\documentclass[12pt]{article}
\usepackage{amsmath,amsthm}
\usepackage{amsfonts}
\usepackage{a4}

\newtheorem{theorem}{Theorem}[section]
\newtheorem{lemma}[theorem]{Lemma}
\newtheorem{proposition}[theorem]{Proposition}
\newtheorem{corollary}[theorem]{Corollary}
\newtheorem{definition}{Definition}[section] 

\theoremstyle{remark}
\newtheorem*{rmk}{Remark}

\newcommand{\Sc}{{\cal S}} 
\newcommand{\TT}{{\cal G}}

\newcommand{\BB}{{\cal B}} 
\newcommand{\U}{{\cal U}}

\newcommand{\M}{{\cal M}}

\newcommand{\ts}{\sigma}

\newcommand{\bm}[1]{\mbox{\boldmath $#1$}}

\newcommand{\Der}{\delta}
\newcommand{\thl}{{\theta}}
\newcommand{\tht}{{\theta_t}}
\newcommand{\thk}{\theta_k}
\newcommand{\eps}{\epsilon} 
\newcommand{\half}{\frac{1}{2}}
\newcommand{\la}{\langle} 
\newcommand{\ra}{\rangle}

\newcommand{\A}{a}
\newcommand{\B}{b}
\newcommand{\C}{u}

\newcommand{\hatq}{{\hat q} {}}

\newcommand{\PP}{\mathbb P}
\newcommand{\QQ}{\mathbb Q}
\newcommand{\Id}{\mathbb I} 
\newcommand{\Real}{\text{Re}}

\def\Journal#1#2#3#4{{#1} {\bf #2}, #3 (#4)}

\def\CQG{\em Class. Quantum Grav.}

\def\PRD{{\em Phys. Rev.} \bm{D}}

\def\JMP{\em J. Math. Phys.}

\def\CMP{\em Commun. Math. Phys.}

\def\PRL{\em Phys. Rev. Lett.}

\def\CPAM{\em Commun. Pure Appl. Math.}
\def\PJM{\em Pacific J. Math.}
\def\ATMP{\em Adv. Theor. Math. Phys.}

\begin{document}

\title{Stability of marginally outer trapped surfaces and existence of 
marginally outer trapped tubes}

\author{Lars Andersson\footnotemark[1]\,\,\footnotemark[2]\,\,,
Marc Mars\footnotemark[3]\,
and
Walter Simon\footnotemark[3] \\
\footnotemark[1] Albert
Einstein Institute, Am M\"uhlenberg 1, D-14476 Potsdam, \\ Germany, \\
\footnotemark[2] Dept. of Mathematics, University of Miami, Coral Gables, \\ FL
33124, USA,
\\
\footnotemark[3] Facultad de Ciencias, Universidad de Salamanca\\
Plaza de la Merced s/n, E-37008 Salamanca, Spain.}


\maketitle

\begin{abstract}
The present work extends our short communication \cite{AMS}.  For smooth
marginally outer trapped surfaces (MOTS) in a smooth spacetime we define
stability 
with respect to variations along
arbitrary vectors $v$ normal to the MOTS. After giving some introductory
material about  linear non self-adjoint elliptic operators,  we introduce the
stability operator $L_v$ and we characterize 
stable MOTS  in terms
of sign conditions on the principal  eigenvalue of $L_v$.
The main result shows that given a strictly stable MOTS $\Sc_0 \subset \Sigma_0$
in a spacetime with a reference foliation $\Sigma_t$, there is an open
marginally outer trapped tube (MOTT), adapted to the reference foliation,
which contains $\Sc_0$. We give conditions under which the MOTT can be
completed.
%
Finally, we show that under standard  energy conditions on the spacetime,
the MOTT 
must be either locally achronal, spacelike or null.
\end{abstract}

\section{Introduction}
\label{Intro}
The singularity theorems of Hawking and Penrose assert the presence of
incomplete geodesics in the time evolution of Cauchy data with physically
reasonable matter containing a trapped surface \cite{HE}.
Studying the structure of these singularities and understanding
whether generically they 
entail a loss of predictability power of the theory have become
central issues in Classical General Relativity. 
More specifically,
the weak cosmic censorship conjecture asserts, roughly speaking, that in the
asymptotically flat case the singularity will be generically
hidden from infinity 
by an event horizon and that a black hole will form. Since this conjecture aims
precisely at showing that a black hole forms,
any sensible approach to its proof should not make 
strong a priori assumptions on the global structure of the spacetime. It is therefore
necessary to replace to concept of black hole, which requires full knowledge
of the future evolution of a spacetime, with a quasi-local concept
that captures its main features and that can be used as a tool to show 
the existence of a black hole.

This approach has been successfully applied in
spherically symmetric spacetimes with matter fields, where the existence
of a complete future null infinity and an event horizon
can in fact be inferred from the
presence of at least one trapped surface in the data \cite{MD}, plus some 
extra assumptions. An important tool in the analysis is to study the 
sequence of the marginally trapped surfaces bounding the region with trapped 
surfaces within each slice of a spherically symmetric foliation.

Quasi-local versions of black holes are important 
not only in the context of cosmic censorship, they are relevant in
any physical situation involving black holes where no global
knowledge of the spacetime is available. An outstanding example
is the dynamics and evolution of black holes. 
In the strong field regime, these
evolutions are so complex that they can only be approached with the aid
of numerical methods. In most cases, numerical computations can only evolve the
spacetime a finite amount of time, which makes the
global definition of black hole of little practical use. A quasilocal 
definition becomes necessary even to define what is understood by a black hole
in this context. More importantly, such a definition is crucial
in order to be able to track the location of the black holes
and to extract relevant physical information from their evolution. Over the
last few years, marginally outer trapped surfaces 
and the hypersurfaces in spacetime which they sweep out during a time
evolution, 
have become standard as quasilocal
replacements of black holes and have been studied extensively, using both 
numerical methods (see, e.g. \cite{BBGB,SK,SKB}) 
as well as analytically, with either
mathematical \cite{AG} or more physical scope \cite{GJ}. 
See \cite{A} for a review of 
some of these issues. However, many open problems still remain. Only in
spherical symmetry a rather complete picture has been obtained 
thanks to \cite{MD,BI,CW}. In general even finding
examples is not easy, see however \cite{DJK,MK,IBD} for interesting
non-spherical examples. 

A marginally outer trapped surface (MOTS)  is a spacelike surface of
codimension two, such that the null expansion $\theta$ with respect to the
outgoing null normal vanishes.  The notions of ``outer'' and ``outgoing'' are
simply defined by the choice of a null section in the two dimensional normal
bundle of the surface.  We call a 3-surface foliated by MOTS  a 
marginally outer trapped tube (MOTT) \cite{AK1,AK2}; an alternative
terminology adopted in \cite{AMS} is ``trapping horizon'', c.f. \cite{SHay}.

The wide range of applicability of MOTS motivated us to study their
propagation in spacetime from an analytical point of view and in a general
context, i.e. assuming neither symmetries nor the presence of any  trapped
regions a priori.  In the context of the initial value problem in general
relativity, namely in a smooth spacetime foliated by smooth hypersurfaces
$\Sigma_t$,  it is natural to ask the following:  Given a MOTS $\Sc_0$ on
some initial leaf ${\Sigma_0}$, does it ``propagate'' to the adjacent leaves
$\Sigma_t$ of the foliation?  In other words, is there a marginally outer
trapped tube starting at $\Sc_0$  whose marginally outer trapped leaves lie
in the time slices $\Sigma_t$?

It turns out that the key property of a MOTS $\Sc_0$ relevant for this
question is its ``stability'' with respect  to the initial leaf
$\Sigma_0$. This concept has interesting applications even when considered 
purely inside a hypersurface $\Sigma_0$, 
in particular for the topology of $\Sc_0$, and also for the
property of being a ``barrier'' for weakly outer trapped and weakly outer
untrapped surfaces (defined by $\theta \le 0$ and $\theta \ge 0$, respectively). We
shall discuss these two issues, which both originate in the work of Hawking
\cite{SH1,SH2}, before turning to the question of propagation of $\Sc_0$ off
$\Sigma_0$.   Hawking's analysis was extended by Newman \cite{RN} who
calculated the general variation $\delta_{v}$  of the expansion $\theta$ with
respect to any transversal direction $v$. A central issue in these papers was
to show that stable MOTS have spherical topology in the generic case.  The
classification of the ``rigidity case'', in which the torus is allowed, was
investigated first for minimal  surfaces \cite{CG1} and subsequently also for
generic MOTS and in higher dimensions \cite{CG2,GS,GG}; see \cite{LA} for a
review.  

A key tool, both for the topological issues as well as for the present
purposes, is the linear elliptic  stability operator $L_{v}$ defined by
$\delta_{\psi v} \theta = L_{v} \psi$ for MOTS (introduced in \cite{AMS},  in deformation
form already present in \cite{RN,CG2}) where $v$ is a suitably
scaled vector.  $L_v$ is  not self-adjoint in general, except in special
cases as for example when the MOTS lies in a time symmetric
slice. Nevertheless, linear elliptic operators always have a real principal
eigenvalue,  and the corresponding principal eigenfunction can be chosen to
be positive \cite{DV1,DV2,Smoller,LE}.  While we define stability and strict
stability of MOTS in terms of sign conditions for preferred variations
$\delta_v \theta$, we can show that this definition is equivalent to
requiring that the principal eigenvalue of  $L_v$ is non-negative or
positive, respectively.  In \cite{AMS} we also showed that strictly stable
MOTS are barriers for all weakly outer trapped and weakly outer  untrapped
surfaces in a neighbourhood. In addition to the properties of the stability
operator,  this result uses a representation of MOTS as graphs over some
reference 2-surfaces. In terms of this representation,  the condition of
vanishing expansion $\theta$ characterizing a MOTS becomes a quasilinear
elliptic equation  for the graph function, for which a maximum principle
holds. 
A similar application of the  maximum principle for the functional $\theta$  is
contained  in the uniqueness results of Ashtekar and Galloway \cite{AG} where
null hypersurfaces through a given MOTS and their intersection with a given
spacelike MOTT were considered.

A barrier property of trapped and marginally trapped surfaces, which
complements the one discussed above has been considered by Kriele and Hayward
\cite{KH}. They showed that the boundary of the trapped region, i.e. the
set of all points in a spacelike hypersurface contained in a bounding,
trapped surface  is a MOTS, under the assumption that it is 
{\em piecewise smooth}. By \cite{AMS}, this MOTS is necessarily stable. 
Andersson and Metzger \cite{AM2} recently showed that the boundary of the trapped 
region is a smooth, embedded, stable MOTS, without any
additional smoothness assumption. 

Turning now to our main problem, namely the propagation of the MOTS $\Sc_0$
off $\Sigma_0$ to adjacent slices, we can prove this for some open time
interval if the MOTS are strictly stable (c.f. Theorem \ref{ex1}). As a tool
we first extend the graph representation of $\Sc_0 \subset \Sigma_0$ to
2-surfaces $\Sc_t  \subset \Sigma_t$.  The linearization of the expansion
operator  $\theta$ is precisely  the stability operator, and strict stability guarantees that
this operator is invertible.  As $\Sc_0$ is a solution of the equation
$\theta = 0$ on $\Sigma_0$, we could apply in our earlier paper \cite{AMS}
the implicit  function theorem to get solutions of $\theta = 0$ in a
neighbourhood.  Here we cut this procedure short by using standard results on
perturbations of differential operators \cite{ADN, AB}  whose linearizations
are elliptic and invertible. Naturally, these results also make use of the
implicit  function theorem. As an easy corollary to  Theorem \ref{ex1}, we find
that the MOTT constructed in this interval is {\it nowhere} tangent to the
$\Sigma_t$. Much more subtle results are  Theorem \ref{thm:lars-ex2}
and Corollary \ref{tang}. Under some genericity condition they show in
particular that, if $\Sc_t$ converge smoothly to a limiting MOTS 
$\Sc_{\tau} \subset \Sigma_{\tau}$ whose principal eigenvalue $\lambda_{\tau}$
vanishes, the resulting MOTT is {\it everywhere} tangent to $\Sigma_{\tau}$. 

We wish to stress that the existence result contained in \cite{AMS} is local
in time. The attempt to formulate  a global result  in \cite{AMS} required
the {\it assumptions} (implicit in the definitions of that paper)  that the
MOTS remain compact, embedded, smooth and strictly stable during the
evolution.  In the present work we allow immersed rather than embedded MOTS
in our existence theorems for MOTT, but we have to deal with the other
potential pathologies.  
To do so we apply a recent result by Andersson and
Metzger \cite{AM1} which shows that, in four dimensional spacetimes,
a sequence of  stable and smooth MOTS which lie in a compact set 
such that area of the sequence stays bounded, converges to a smooth and stable MOTS.
This leads to Proposition \ref{prop:lars-ex2}. 
Moreover, if the dominant energy condition holds, it is
easy to show that the area stays in fact bounded provided  the MOTS remain
strictly stable in the limit.  This gives Theorem \ref{ex3} as a sharper version of
our existence result.

This paper is organized as follows. In section \ref{basic} we explain the
most important items of our notation. In
section \ref{variation} we discuss the variation of the expansion, and 
introduce the stability operator. The somewhat technical computation
of the variation, which simplifies the derivation by Newman, 
is given in  appendix \ref{sec:defproof}. 

We proceed with some technical material on linear elliptic operators
with first order term, cf. section  \ref {elliptic}. 
Here and in appendix \ref{sec:appendB} 
we give an exposition of,  in particular,  the
Krein-Rutman theorem \cite{KR} on the principal eigenvalue and
eigenfunction of linear elliptic operators, 
and the maximum principle for  operators with non-negative
principal eigenvalue \cite{BNV}.  We continue in section  \ref{stability} by
discussing in detail stability definitions for MOTS, in particular the
relation between the variational definitions and the sign  condition on the
principal eigenvalue, and we give a result on the dependence of stability on
the direction.  Sect. \ref{graph} contains the graph representation of a
MOTS.  In Sect. \ref{barrier} we describe the barrier properties of MOTS
which satisfy suitable stability conditions, along the lines sketched above,
slightly extending our earlier paper \cite{AMS}.  In section
\ref{symmetries}, we show, roughly speaking, that
strictly stable MOTS inherit the symmetry of the ambient geometry,  and that 
the
same is true for the principal eigenfunction.  In the final section
\ref{mott} we prove existence for MOTT in the three theorems \ref{ex1},
\ref{thm:lars-ex2} and \ref{ex3} already
sketched above. 
Our final Theorem \ref{achronal} is a slight extension of a result of \cite{AMS}
and shows that under standard energy conditions, suitable
(non-)degeneracy conditions for the initial MOTS $\Sc_0$ and 
for spacelike or null reference foliations, the MOTT through
$\Sc_0$  is either spacelike or null everywhere on $\Sc_0$.

\section{Some basic definitions}
\label{basic}
A spacetime $(\M,g)$ is an $n$-dimensional oriented and time-oriented
Hausdorff manifold endowed with a smooth metric of Lorentzian signature
$+2$. Some results below require $n=4$.  $\Sc$ will denote an orientable,
closed (i.e. compact without boundary) codimension 2, immersed submanifold of
$\M$ with positive definite first fundamental form $h$. 
An object with all
these properties is simply called ``surface'' throughout this paper. 
The area of $\Sc$ will be denoted by $|\Sc|$. 
Spacetime
tensors will carry Greek indices and tensors in $\Sc$ will carry
capital Latin indices. Our
conventions for the second fundamental form (-vector) and the mean curvature
(-vector) are  $K(X,Y) \equiv - ( \nabla_{X} Y)^{\bot}$ and $H = \mbox{tr}_h
K$. Here $X$, $Y$ are tangent  vectors to $\Sc$, ${\bot}$ denotes the
component normal to $\Sc$ and $\nabla$ is the Levi-Civita connection  of
$g$. $\Sc$ will always be assumed to be smooth unless otherwise stated.

The normal bundle of $\Sc$ is a Lorentzian vector bundle which admits a null
basis $\{l^{\alpha},k^{\alpha}\}$ which we always take future directed,
smooth and normalized  so that $l^{\alpha} k_{\alpha}=-2$. This basis is
defined up to a boost $l^{\alpha} \rightarrow \kappa l^{\alpha}$, $k^{\alpha}
\rightarrow \kappa^{-1} k^{\alpha} $, $\kappa>0$.  The {\it null expansions}
of $\Sc$ are $\theta_l =  H^{\alpha}l_{\alpha}$ and $\thk =
H^{\alpha}k_{\alpha}$ and the mean  curvature in this null basis reads $
H^{\alpha}= - \half \left ( \thk l^{\alpha} + \theta_l k^{\alpha} \right)$.

\begin{definition}
A surface $\Sc$ is a {\bf marginally outer trapped surface} (MOTS) if
$H^{\alpha}$ is proportional to one of the  elements of the null basis of its
normal bundle. 
\end{definition}
This condition is more restrictive than just demanding $H^{\alpha}$ to be
null because it excludes the possibility  that $H^{\alpha}$ points along
$l^{\alpha}$ in some open set and along $k^{\alpha}$ on its complement (c.f.
\cite{MSen}).  The
null vector to which $H^{\alpha}$ is proportional is called $l^{\alpha}$, and
the direction to which  it points is called the {\em outer} direction 
(in the case
$H^{\alpha} \equiv 0$, both $l^{\alpha}$ and $k^{\alpha}$ are outer 
directions).  In other words, the term outer does not refer to a direction
singled out {\it a priori}, but to the fact that we only have information
about, or we are only interested in,  {\it one of} the expansions.
Equivalently, $\Sc$ is a MOTS iff $\theta_l=0$.  If furthermore $H^{\alpha}$
is either future or past directed everywhere $\Sc$ is called {\it marginally
trapped}.  Hence a marginally trapped surface satisfies $\theta_l = 0$ and
either $\thk \leq 0$ or $\thk \geq 0$ everywhere.  Next, $\Sc$ is called {\it
weakly outer trapped} iff at least one of the expansions in non-positive,
say $\theta_l \leq 0$. {\it  Weakly outer untrapped} surfaces satisfy the 
reverse inequality.  Finally, in order for $\Sc$ to be a future (past)
{\it trapped
surface} we require that the strict inequalities $\theta_l < 0$ and $\thk <
0$, ($\theta_l > 0$, $\thk > 0$) hold. 
Since we will only deal with $\theta_l$ from now on, we simplify
the notation and refer to it simply as $\theta$.

A  marginally outer trapped tube $\TT$ is a smooth collection of MOTS.  More
precisely, we state the following
\begin{definition}
Let $I \subset \mathbb{R}$ be an interval.  A hypersurface $\TT$, possibly
with boundary, is a 
{\bf marginally outer trapped tube} (MOTT) if there 
is a smooth immersion $\Phi: \Sc \times I \rightarrow M$, such
that $\TT = \Phi(\Sc\times I)$ and 
\begin{enumerate}
\item for fixed $s \in I$, $\Phi(\Sc,s)$ is a MOTS with respect to a smooth
field of null normals $l^{\alpha}$ on $\TT$ 
and 
\item \label{point:phistar}
$\Phi_{\star}
(\partial_s)$ is nowhere zero.
\end{enumerate}
Suppose $(\M,g)$ contains a foliation by hypersurfaces
$\{\Sigma_t\}_{t \in J}$. 
A MOTT $\TT$ is said to be {\bf adapted} to the reference
foliation $\{\Sigma_t\}$ if for each $s \in I$, it holds that 
$\Sc_{\sigma(s)} = \Phi(\Sc,s)$ is a MOTS in $\Sigma_{\sigma(s)}$, 
where $\sigma: I \to J$ is a smooth,
strictly monotone function.
\end{definition}
\begin{rmk} Note that if $I$ contains at least 
one of its boundary points, then the
MOTT $\TT$ is a hypersurface with boundary. 
\end{rmk}
Condition \ref{point:phistar} is required in order to allow self-intersections
of the MOTS
but not in the direction of propagation. For embedded MOTS, the MOTT is also
embedded and its definition is equivalent to Hayward's ``trapping horizons''
\cite{SH1}. The terms ``dynamical horizon'' and ``isolated horizon''
\cite{AK1,AK2} are particular cases in which the causal structure is
restricted {\it a priori} .

\section{Varying the expansion}
\label{variation}
A fundamental ingredient in our existence theorem is the first order
variation of the vanishing null expansion $\thl$ of a MOTS. This variation
was given in full generality by Newman in \cite{RN} for arbitrary immersed
surfaces (not necessarily MOTS).  We give here a simplified derivation.  We
first have to introduce some notation.

Let $\nabla_{\alpha}$ and $G_{\alpha\beta}$ denote the covariant derivative
and Einstein tensor of $(\M,g)$ respectively. Let  $(\Sc,h)$ be a spacelike
codimension-two surface  immersed in $\M$ (not necessarily closed),
$e^{\alpha}_{A}$ any basis of the tangent space of $\Sc$, $D_A$  the
covariant derivative on $(\Sc,h)$ and $R_{\Sc}$ its curvature scalar.  Let us
fix a null basis $\{l^{\alpha},k^{\alpha} \}$ in the normal bundle of $\Sc$
and define a one-form on $\Sc$ as
$$
s_A = - \frac{1}{2}  k_{\alpha} \nabla_{e_A} l^{\alpha}.
$$
We shall calculate how $\thl$ changes when the surface $\Sc$ is varied
arbitrarily.  This variation is defined by an arbitrary spacetime vector
$q^{\alpha}$ defined along $\Sc$.  More specifically, let $0 \in I \subset
\mathbb{R}$ be an open interval and $\Phi: \Sc \times I \rightarrow \M$ be a
differentiable map such that for fixed $\sigma$, $\Phi(\cdot,\sigma)$ is an
immersion and for fixed $p$, $x_p^{\alpha} (\ts) \equiv \Phi^{\alpha}(p,\ts)$
is a curve starting at $p \in \Sc$ with tangent vector $q^{\alpha}(p)$.
Define the family of two-surfaces $\Sc_{\ts} \equiv \Phi (\Sc, \ts)$. 
Let
$l_{\ts}^{\alpha}$ be a nowhere zero null vector on the normal  bundle of
$\Sc_{\ts}$ which is differentiable with respect to $\ts$, and let
$\thl_{\ts}$ be the null expansion of each surface $\Sc_{\ts}$.  The
variation of $\thl_{\ts}$, defined as $\Der_{q} \thl \equiv \partial_{\ts}
\thl_{\ts}  |_{\ts=0}$, depends only on $q^{\alpha}$ and on the null vector
field $l^{\alpha}$ and its first variation (if $\Sc$ is a MOTS this last
dependence also drops out), but not on the details  of the map $\Phi$. For a
MOTS, the variation is moreover linear in the sense that
\begin{equation}
\label{linvar}
\Der_{a q_1  + q_2} \thl = a \Der_{q_1} \thl +  \Der_{q_2} \thl
\end{equation}
for any {\it constant} $a$, while in general $\Der_{\psi q} \thl \neq \psi \Der_{q}
\thl$ for functions $\psi$.

It should be noted that  in the context of trapping and dynamical horizons,
derivatives of $\thl$ have been  employed frequently (for instance in the
definition of outer trapping horizon by Hayward \cite{SHay}  or in the
uniqueness results by Ashtekar and Galloway \cite{AG}). These are {\it not}
the variations  we are considering in this paper. The former are derivatives
of a scalar function defined in a neighbourhood of the  horizon by extending
$\thl$ off the horizon, using the Raychaudhuri equation. In this case, the
derivative is obviously  linear with respect to multiplication by scalar
functions, unlike the geometric variation employed here.

The variation vector $q^{\alpha}$ can  be decomposed into a tangential part
$q^{\,\|\alpha}$ and an orthogonal part $q^{\bot\alpha}$   with respect to
$\Sc$, i.e.  $q^{\alpha}=q^{\bot\alpha} + q^{\,\|\alpha}$. The normal
component can in turn be decomposed in terms of the null basis as
$q^{\bot\alpha} = \B l^{\alpha} - \frac{\C}{2} k^{\alpha}$ where $\B$, $\C$ are
functions on $\Sc$. The following result, whose proof we have shifted to 
appendix \ref{sec:defproof}, 
gives the explicit expression for the variation of $\thl$ along
$q^{\alpha}$.

\begin{lemma}
\label{gv} 
Let $\Sc$, $l^{\alpha}$,  $\thl$ and $q^{\alpha}= {q}^{\, \| \alpha}+ \B
{l}^{\alpha} - \frac{\C}{2} {k}^{\alpha}$ be as before. Then, the variation of
$\thl$ along ${q}$ is
\begin{align*} 
\delta_{{q}} \thl &= \A \thl + {q}^{\, \|} \left ( \thl \right ) - \B \left (
K^{\mu}_{AB} K^{\nu \, AB} l_{\mu} l_{\nu} + G_{\mu\nu} l^{\mu} l^{\nu}
\right ) - \Delta_\Sc \C + 2 s^A D_A \C 
\\  
&\quad + \frac{\C}{2}  \left (
R_\Sc - {H}^2 - G_{\mu\nu} l^{\mu} k^\nu - 2 s^A s_A + 2 D_A s^A \right ),
\end{align*} 
where $\Delta_{\Sc} = D_A D^A$ is the Laplacian on $(\Sc,h)$ and $\A =  -
\frac{1}{2} k_{\alpha} \partial_{\ts} l^{\alpha}_{\ts} |_{\ts =0}$.
\end{lemma} 

The decomposition of $q^{\bot\alpha}$ in the null basis
$\{l^{\alpha},k^{\alpha} \}$ is natural for  the surface $\Sc$ as a
codimension-two submanifold in spacetime, and Lemma \ref{gv} gives the
general variation with respect to arbitrary vectors on $\Sc$.  However, we
will later refer the variations of $\Sc$  to some foliation of $\M$ by
hypersurfaces  $\Sigma_t$ with $\Sc \subset \Sigma_0$ (c.f. section
\ref{barrier}), and for this we will employ a natural  alternative
decomposition of $q^{\bot\alpha}$ adapted to the foliation. We consider an
arbitrary normal vector field $v^{\alpha}$ to $\Sc$ which is, at each point,
linearly independent of $l^{\alpha}$, and we impose the normalization
$v^{\alpha} l_{\alpha} = 1$  which does not restrict the causal character of
$v^{\alpha}$ anywhere on $\Sc$.  $v^{\alpha}$ is uniquely defined by a scalar
function $V$ on  $\Sc$ according to
\begin{equation}
v^{\alpha} = - \half k^{\alpha} + V l^{\alpha}. \label{vectorv}
\end{equation}

We use $\{v^{\alpha},l^{\alpha}\}$ as a basis in the normal space.  Inverting
(\ref{vectorv})  we get $k^{\alpha} = 2 ( V l^{\alpha} - v^{\alpha} )$.  We
next define a vector $u^{\alpha} = \half k^{\alpha} + V l^{\alpha}$,  which
is orthogonal to $v^{\alpha}$ and satisfies  $u^{\alpha} u_{\alpha} = -
v^{\alpha} v_{\alpha} = - 2 V$.  The quantities
\begin{align}
W  &=  K^{\mu}_{AB} K^{\nu \, AB} l_{\mu} l_{\nu}  + G_{\mu\nu} l^{\mu}
l^{\nu}, \label{EqW} \\  
Y  &=   V K^{\mu}_{AB} K^{\nu \, AB} l_{\mu} l_{\nu}
+  G_{\mu \nu} l^{\mu} u^{\nu}
\label{eqY}
\end{align}
will appear frequently below. Clearly 
$W$ is non-negative provided the null energy
condition holds and $Y$ is non-negative if $u^{\alpha}$ is causal and the
dominant energy condition holds.

The following definition introduces an object which plays a key role in this
paper. 
\begin{definition}
The {\it stability operator}  is defined by
\begin{equation}
L_{v} \psi   =  - \Delta_\Sc \psi + 2 s^A D_A \psi + \left ( \half R_\Sc  -
Y  -  s^A s_A +  D_A s^A  \right ) \psi.
\label{Lv}
\end{equation}
\end{definition}
The following lemma is a trivial specialization of 
Lemma \ref{gv}.
\begin{lemma}
\label{varth} 
Let $\Sc$ be a MOTS, i.e.  $\thl = 0$.  The variation of the expansion $\thl$
on $\Sc$ with respect to the null vector $\psi l^{\alpha}$,  and with respect
to any vector $\psi v^{\alpha}$ orthogonal to $\Sc$ with
$l^{\alpha}v_{\alpha} = 1$, respectively, are given by
\begin{align}
\label{Ray} 
\delta_{\psi l} \thl & =  - \psi  W, \\
\label{varv}
\delta_{\psi v} \thl & =  L_v \psi,
\end{align}
\end{lemma}
We note that (\ref{Ray}) is the Raychaudhuri equation.
For arbitrary vectors $w^{\alpha}$ orthogonal to $\Sc$ and linearly
independent of $l^{\alpha}$, not necessarily normalized to satisfy
$w^{\alpha} l_{\alpha} =1$, we can define another elliptic operator $L_{w}
\psi = \delta_{\psi w} \thl$. In terms of the normalized vector $v \equiv
F^{-1} w$, where $F = w^{\alpha} l_{\alpha}$ we obviously have
\begin{equation}
\label{resc}
L_{w} \psi = L_{Fv} \psi = L_{v} (F \psi ).
\end{equation}
Hence, $L_v$ and $L_{Fv}$ contain essentially the same information.
To see the  dependence of $L_v$ on the vector $v$, i.e. on the function $V$
we calculate, from linearity (\ref{linvar}), and (\ref{Ray}),
\begin{equation}
L_v = \half L_{- k} - V W.
\label{LvLk}
\end{equation}

Due to the presence of the first order term in (\ref{Lv}), $L_v$ is not
self-adjoint in general.  However, self-adjoint extensions still exist in
special cases. For example, if  $s_A$ is a gradient,  i.e. $s_A = D_A \zeta$
for some $\zeta$, $L_v$ is self-adjoint with respect to a suitable measure
depending on $\zeta$, c.f.  section \ref{elliptic}.

Since $L_v$ is linear and elliptic, it has discrete eigenvalues and 
the corresponding eigenfunctions are regular. However, in general, the
eigenvalues and eigenfunctions are complex. 
Nevertheless, the principal eigenvalue, i.e. the 
eigenvalue with smallest real part, and its
corresponding eigenfunction behave in the same manner as for self-adjoint
operators. In particular, they can be used to give 
a very efficient reformulation of the maximum
principle.  In the next section we collect some material on linear elliptic
operators which will be the key tools in the subsequent discussion of
stability.

\section{Properties of linear elliptic operators}
\label{elliptic}
The results of this section hold for connected compact smooth manifolds $\Sc$
without boundary, and for arbitrary smooth, linear, second order, elliptic
operators on $\Sc$, which can be written as
\begin{equation}
\label{ell}
L = - \Delta_h + 2 t^A D_A + c
\end{equation}
where  $\Delta_h \psi= D_A \left(h^{AB} D_B \psi \right)$,
$h^{AB}$ is positive definite and smooth, $D_A$ is the 
corresponding Levi-Civita covariant derivative 
and  $t^A$ and $c$ are smooth.

\begin{lemma}  
\label{principal}
Let $L$ be an elliptic operator of the form (\ref{ell}) on a compact
manifold. Then, the following holds. 
\begin{enumerate} 
\item There is a real eigenvalue $\lambda$, called the principal eigenvalue,
such that for any other eigenvalue $\mu$ the inequality $\mbox{Re}
(\mu) \geq \lambda$ holds.  The corresponding eigenfunction $\phi$, $L
\phi= \lambda \phi$ is unique up to a multiplicative constant and can be
chosen to be real and everywhere positive.
\item \label{point:princ-ii} 
The adjoint $L^{\dagger}$ (with respect to the $L^2$ inner
product)  has the same principal eigenvalue $\lambda$ as $L$.
\end{enumerate}
\end{lemma} 
Applications of Lemma \ref{principal} to the stability operator (\ref{Lv})
in a spacetime will be described in the next section. 
Note that our terminology 
follows Evans \cite{LE}; in particular, regarding the sign of
$L$ it is  opposite to other references \cite{DV1,BNV} cited below.  This
entails that when comparing with these references one always has to
interchange ``sup'' and ``inf'' when acting on expressions containing $L$.

The existence of the principal eigenvalue and eigenfunction was first proven
by Donsker and  Varadhan \cite{DV1} using parabolic theory and the
Krein-Rutman theorem \cite{KR}.  For uniqueness of the principal
eigenfunction, c.f. Berestycki, Nirenberg and Varadhan \cite{BNV}.  All these
papers actually deal with the Dirichlet problem for bounded domains in
$\mathbb{R}^n$.  However, the proof is easily adapted  to the case of compact
connected manifolds without boundary (it is in fact simpler).  For
completeness and since this result is not widely known, we provide in
appendix \ref{sec:appendB} a sketch of the proof in the closed manifold case.
The sketch
follows the argument in Smoller \cite{Smoller} and Evans \cite{LE}.

For self-adjoint operators $L_0 = - \Delta_h + c$, 
the principal eigenvalue
$\lambda_0$ is given by the Rayleigh-Ritz formula 
\begin{equation}
\lambda_0 = \inf_u { \la u, L_0 u \ra}=  \inf_u \int_\Sc \left (D_A u D^A u +
c u^2 \right ) \eta_\Sc, \label{RR}
\end{equation} 
where $\eta_\Sc$ is the surface element on $(\Sc,h)$
and the infimum is taken
 over smooth functions $u$ on $\Sc$ with $||u||_{L^2} = 1$.

For non-self-adjoint operators such a characterization is no longer true.
However, Donsker and Varadhan \cite{DV1} have given  alternative variational
representations of the eigenvalue, namely
\begin{align}
\label{var1}
\lambda  & =   \inf_{\mu_{\Sc}} \sup_{\psi} \int_{\Sc} \psi^{-1}{L \psi} \mu_\Sc,\\
\label{var2}
\lambda & =  \sup_{\psi} \inf_{x \in \Sc} \, \psi^{-1}(x){L \psi(x)}.
\end{align}
Here the infimum in (\ref{var1}) is taken over all probability measures
$\mu_\Sc$ on $\Sc$, while the suprema are over all smooth, positive functions
$\psi$ on $\Sc$.

To get (\ref{var2}) from (\ref{var1}) we first note that the ``inf'' and 
``sup'' in  (\ref{var1}) can be interchanged (which is non-trivial but follows
from a min-max theorem of Sion \cite{MS}, c.f. \cite{DV1}). This done, the
infimum  of the integral with respect to all probability measures is achieved
by a Dirac delta measure concentrated at the location where the integrand
takes its infimum.

In order to approach a characterization closer to a Rayleigh-Ritz expression,
we note that, since probability measures can be approximated by smooth
positive functions, we can assume $\mu_\Sc = u^2 \eta_\Sc$ with smooth $u >
0$ and $||u ||_{L^2} = 1$.  Following Donsker and Varadhan, $\psi$ can be
decomposed as $\psi = u e^{\omega}$. Direct substitution in (\ref{var1}) and
rearrangement gives
$$
\lambda =  \inf_{u} \sup_{\omega} \int_\Sc \left ( D_A u D^A u +  \left ( c +
t_A t^A - D_A t^A \right ) u^2  -  \left ( D_A \omega + t_A \right )^2
u^2\right ) \eta_\Sc.
$$
To reformulate this expression we use the Hodge decomposition $t_A =  D_A
\zeta + z_A$, where $\zeta$ is a function and $z_A $ is divergence-free.
This decomposition is unique except for a constant additive term in $\zeta$.
The supremum over $\omega$ only affects the last term, and it only depends on
$z_A$ and not on  $\zeta$. Thus, we need to determine  $\inf_{\omega}
\int_\Sc \left ( D_A \omega + z_A \right )^2 u^2 \eta_\Sc$, for each $u$.  A
standard argument shows that the infimum is achieved and is given by the
solution of the linear elliptic equation
\begin{equation}
- D_A D^A \omega[u] - 2u^{-2} D_A \omega[u] D^A u = 2 u^{-2} z^A D_A u,
\label{eqom}
\end{equation}
where we write $\omega[u]$ to emphasize that the solution depends on $u$.  A
trivial Fredholm argument shows that this equation has a unique solution
satisfying
\begin{equation}
\int_\Sc u^2 \omega[u] \eta_\Sc = 0. \label{intCond}
\end{equation}
It therefore follows that the Donsker-Varadhan characterization of the
principal eigenvalue can be rewritten as
\begin{equation}
\lambda = \inf_{u} \int_\Sc \left ( D_A u D^A u + Q u^2 - \left ( D_A
\omega[u] + z_A \right )^2 u^2 \right ) \eta_\Sc,
\label{DV2}
\end{equation}
where $Q = c + t_A t^A - D_A t^A$ and the infimum is taken over smooth
functions of unit $L^2$ norm. The symmetrized operator $L_s = - \Delta_h +
Q$ is self-adjoint and has a principal eigenvalue $\lambda_s$ given by the
Rayleigh-Ritz formula, as in (\ref{RR}). Since the last term in (\ref{DV2})
is non-positive, the inequality 
\begin{equation}\label{eq:lams}
\lambda_s \geq \lambda
\end{equation} 
follows immediately.
This inequality has recently been demonstrated 
by Galloway and Schoen \cite{GS} using
direct estimates. It is interesting that (\ref{eq:lams}) is a trivial
consequence of the 
Rayleigh-Ritz type formula (\ref{DV2}) 
for the principal eigenvalue.

As a second application of the characterization (\ref{DV2}),  we note that the
last term can be rewritten as
$$
- \left ( D_A \omega [u] + z_A \right )^2 u^2  = \left ( D_A \omega[u] D^A
\omega[v] - z_A z^A \right ) u^2  -2 D_A \left ( \omega[u] u^2 \left (z^A +
D^A \omega[u] \right) \right )
$$
after using the equation (\ref{eqom}) and the fact that $z^A$ is
divergence-free.  Integration on $\Sc$ gives the alternative representation
for $\lambda$
\begin{equation}
\lambda = \inf_{u} \int_\Sc \left ( D_A u D^A u + \left ( Q - z_A z^A  \right
) u^2 + u^2 D_A \omega[u] D^A \omega[u] \right ) \eta_\Sc. \label{DV3}
\end{equation}
Since the last term is now non-negative, dropping it decreases the integrand
and therefore also the infimum. Thus, defining the alternative symmetrized
operator $L_z \equiv L_s - z_A z^A$ and denoting by $\lambda_z$ the
corresponding principal eigenvalue, it follows that
$$
\lambda \geq \lambda_z.
$$
Therefore the principal eigenvalue of a non self-adjoint operator is always
sandwiched between  the principal eigenvalues of two canonical symmetrized
elliptic operators, and we also note that $\inf_{\Sc} (z_A z^A)\le \lambda_s
- \lambda_z \le \sup_{\Sc} (z_A z^A)$.  Obviously, when $L$ is self-adjoint
(w.r.t. to the $L^2$ norm) i.e. $t_A \equiv 0$, the two symmetrized operators
$L_s$ and $L_z$ coincide with it.  More generally, when $t _A$ is a gradient
(so that $L$ is self-adjoint with a suitable measure) the characterization of
$\lambda$ given by (\ref{DV2}) coincides  with the Rayleigh-Ritz expression
because whenever $z_A=0$, the unique solution to (\ref{eqom}) and
(\ref{intCond}) is just $\omega[u]=0$ for all $u$.

An important tool in the analysis of the properties of the stability operator
will be the maximum principle.  The textbook formulations normally require
that the coefficient of the zero order term of the elliptic operator  is
non-negative, at least if a source term is present (see for example 
\cite[Sect. 3]{GT}).  We give here a reformulation, used in several places below,
which instead requires non-negativity or positivity  of the principal
eigenvalue.
   
\begin{lemma}
\label{MP}
Let $L$ be a linear elliptic operator of the form (\ref{ell}) on a compact
manifold. Let $\lambda$ and $\phi > 0$ be the principal eigenvalue and
eigenfunction of $L$, respectively, and
let $\psi$ be a smooth solution of $L \psi = f$ for some smooth function  $f
\ge 0$. Then the following holds. 
\begin{enumerate}
\item \label{point:MPi} 
If $\lambda = 0$, then $f \equiv 0$ and $\psi=C \,\phi$ for some constant
$C$.
\item \label{point:MPii} 
If $\lambda > 0$ and $f \not\equiv 0$, then $\psi > 0$.
\item \label{point:MPiii} 
If $\lambda > 0$ and $f \equiv 0$, then $\psi \equiv 0$.
\end{enumerate}
\end{lemma}
\begin{rmk}
Clearly  analogous results hold for $f \le 0$.
\end{rmk}  
\begin{proof}
For a positive principal eigenfunction $\phi$ of $L$,  we define
$\chi$ by $\psi = \chi \phi$  and an operator $\Gamma[\phi]$ by the first
equation in
\begin{equation}
\label{delchi}
\Gamma[\phi] \chi =  - D_A \left (\phi^2 h^{AB} D_A \chi \right) +  2 \phi^2 t^A
D_A \chi + \lambda \phi^2 \chi = \phi f
\end{equation}
while the second equality follows by computation.  The strong maximum
principle \cite[Theorem 3.5]{GT}, applied to  (\ref{delchi}) gives the desired
results. 
\end{proof}

We end this section with a result which is essentially a straightforward
application of linear algebra. 
\begin{lemma} \label{lem:linalg}
Let $L$ be a second order elliptic operator on $\Sc$ of the  form (\ref{ell}). 
Let $\lambda, \phi, \phi^\dagger$ be the real principal eigenvalue, 
and the corresponding real eigenfunctions of $L$ and its adjoint $L^\dagger$.  
Let $\PP$ be the projection operator 
defined by  
$$
\PP f = \phi \frac{\la \phi^\dagger, f \ra} {\la \phi^\dagger, \phi \ra }
$$
and let $\QQ = \Id - \PP$. Then 
\begin{enumerate} 
\item \label{point:lam-i}
$L = \lambda \PP + A$, where $A$ has spectrum $\sigma(A)$ such that for some
  constant $c_0 > 0$, 
\begin{equation}\label{eq:Amubound}
\min_{\mu \in \sigma(A) , \mu \ne 0} | \mu - \lambda |  > c_0 
\end{equation}
\item \label{point:lam-ii}
For any $s \geq 2$, $p > 1$, there is a constant $C$ 
such that the inequality 
\begin{equation}\label{eq:Aest}
|| \QQ u ||_{W^{s,p}} \leq C ||(A - \lambda \QQ) u||_{W^{s-2,p}}
\end{equation}
holds, where $W^{s,p}$ is the usual Sobolev space of functions
with $s$ derivatives in $\mbox{L}^p (\Sc)$.
\end{enumerate} 
\end{lemma} 
\begin{rmk} We refer to the constant $c_0$ as the spectral gap.
\end{rmk} 

\begin{proof} 
We shall consider complex eigenvalues and eigenvectors, 
and use the sesquilinear $L^2$ pairing 
$$
\la u, v \ra = \int_\Sc \bar u v.
$$
It is straightforward to check that $L\PP = \PP L = \lambda \PP$,  and if we
define the operator $A$ by $A = L \QQ$, so that  
$$
L = \lambda \PP + A, 
$$
then 
\begin{equation}\label{eq:Aprop} 
A \PP  = \PP A = 0, \quad A \QQ = \QQ A = A, \quad A^\dagger \PP^\dagger =
 0, \quad 
A^\dagger \QQ^\dagger  = A^\dagger,
\end{equation}  
where we used that  $\PP^{\dagger}$ is the projector onto $\phi^{\dagger}$ i.e.
$\PP^{\dagger} f = \phi^{\dagger}  \la \phi, f \ra / \la \phi^\dagger, \phi \ra $.
It follows that the range of $A$ lies in the space orthogonal to
$\phi^{\dagger}$.

Let $\psi$ be an eigenvector for $L$ corresponding to a non-principal
eigenvalue $\mu \ne \lambda$. We find that 
$$
\mu \la \phi^\dagger, \psi \ra = 
\la \phi^\dagger, L \psi \ra = 
\la L^\dagger \phi^\dagger, \psi \ra = 
\lambda \la \phi^\dagger, \psi \ra. 
$$
Since $\mu \ne \lambda$ by assumption, we conclude that 
$\la \phi^\dagger , \psi \ra = 0$. Combining the previous  facts, it 
follows easily that the spectrum $\sigma(A)$ of $A$ satisfies 
$$
\sigma(A) = ( \sigma(L) \setminus \{\lambda\} ) \cup \{ 0 \}.
$$
Since $\lambda$ is the principal eigenvalue of $L$, so that $\lambda \leq
\Real (\mu)$ for any non-principal eigenvalue $\mu$, 
(\ref{eq:Amubound}) follows from the fact that $L$ has discrete spectrum. 

To prove point \ref{point:lam-ii}, we first note that
by construction, $A - \lambda \QQ$ is a Fredholm
operator which maps the space orthogonal
to $\phi^{\dagger}$ into itself. 
By point \ref{point:lam-i}, 
the spectrum of 
$A - \lambda \QQ$ is bounded from below by $c_0 > 0$. 
Thus, by the Fredholm alternative
$A - \lambda \QQ : \la \phi^{\dagger}
\ra^{\bot} \rightarrow \la \phi^{\dagger}\ra^{\bot}$
is invertible on Sobolev spaces.
\end{proof} 

\section{Stability of MOTS}
\label{stability}

The concept of stability of a marginally outer trapped surface 
with respect to a given slice $\Sigma$ is
a central issue of our paper and crucial for the application to existence of
marginally outer trapped tubes. 
We first briefly comment on stability of minimal surfaces embedded in
Riemannian  manifolds $(\Sigma,\gamma)$, disregarding any embeddings in
spacetimes.  Letting $m^{i}$ be the unit normal to $\Sc$, the stability
operator $L_m (\zeta) \equiv \delta_{\psi m} p$ where $p$ is the mean
curvature of $\Sc$ reads (e.g. \cite{CM}) $L_m \zeta = - \Delta_{\Sc} \zeta -
( R_{ij} m^i m^j + K_{ij} K^{ij} ) \zeta$, where and $R_{ij}$ and $K_{ij}$
are the Ricci tensor of $(\Sigma,\gamma)$ and second fundamental form of
$\Sc$ respectively.  This expression also follows from (\ref{Lv}) by taking
$\Sc$ immersed in a time symmetric spacelike hypersurface $(\Sigma,\gamma)$
and choosing $v = m$.  As $L_m$ is self-adjoint its  principal eigenvalue
$\lambda$ can be represented as the Rayleigh-Ritz formula (\ref{RR}). In
terms of the latter, the standard formula for the variation of the area $A$
of a minimal surface along a vector  $v^{i} = \psi m^{i}$  gives
\begin{equation}
\label{delA}
\delta^2_v A = \delta^2_v \int_\Sc\eta_\Sc =  \int_\Sc\psi \delta_v p
\eta_\Sc = \int_\Sc \psi L_m \psi \eta_\Sc \ge \lambda \int_\Sc \psi^2
\eta_\Sc.
\end{equation}

The minimal surface is called {\it stable} if $\lambda \geq 0$, and
(\ref{delA}) together with (\ref{RR}) shows that this is equivalent to
$\delta^2_v A \ge 0$, for  all smooth variations $\delta_v$.

We wish to characterize stable MOTS embedded in arbitrary hypersurfaces in
spacetime  in a similar way as stable minimal surfaces in Riemannian
manifolds.  In the general case we lose the connection between stability and
the second variation of the area (except if the variation is directed along
$l^{\alpha}$).  Using the discussion in the previous section, we can now
define stability either in terms of a  sign condition on suitable first
variations of $\thl$ at the MOTS or,  in view of (\ref{Lv}) by a sign
condition on the principal eigenvalue of the stability operator.  In either
case, we wish to establish a relation between both concepts.  Here we choose
the variational formulation as definition, from which we show the properties
of the principal eigenvalue.

We could obtain a variational definition by replacing (\ref{RR}) by one of
(\ref{var1}), (\ref{var2}) or (\ref{DV2}) but the resulting definition does
not seem very illuminating. Therefore, instead of defining stability in terms
of sign requirements on {\it integrals over} $\Sc$ for {\it all} positive
variations of $\thl$, we now require the existence of {\it at least one}
variation  along which $\thl$ has a sign {\it everywhere} on $\Sc$.  This
definition also enables us to introduce an  important refinement, namely to
distinguish whether this preferred variation of $\thl$ is just  non-negative
everywhere, or even positive somewhere. (We slightly expand our earlier
presentation \cite{AMS}). 
We recall that $v^{\alpha}$ is linearly independent
of $l^{\alpha}$ everywhere on $\Sc$.

\begin{definition}\label{def:stablyouter} 
A marginally outer trapped surface  $\Sc$ is called {\bf stably outermost}
with respect to the direction $v$ iff there exists a function $\psi \geq 0$,
$\psi \not\equiv 0$, on $\Sc$ such that  $\Der_{\psi v} \thl \geq 0$.  $\Sc$
is called {\bf strictly stably outermost} with respect to the direction v if,
moreover, $\Der_{\psi v} \thl \ne 0$  somewhere on $\Sc$.
\end{definition} 
We will omit the phrase ``with respect to the direction $v$'' when this  is
clear from the context.  We now establish the connection between stability
and the sign of the principal eigenvalue.

\begin{proposition} \label{lem:eigen} 
Let $\Sc \subset \Sigma$ be a MOTS and   let $\lambda$ be the principal
eigenvalue of the corresponding operator $L_v$. Then $\Sc$ is stably
outermost iff $\lambda \geq 0$ and strictly stably outermost iff $\lambda >0$.
\end{proposition}
\begin{proof}
If $\lambda \geq 0 (>0)$, choose $\psi$ in the definition of
(strictly) stably outermost as a positive eigenfunction $\phi$ corresponding
to $\lambda$.  Then $\Der_{\phi v} \thl =  L_{v} \phi = \lambda \phi \geq 0
(>0)$.  For the converse, we note that from Lemma \ref{principal} the
adjoint $L^{\dagger}_v$   (with respect to the standard $L^2$ inner product
$\la \, , \ra$ on $\Sc$) has the same principal eigenvalue as $L_v$, and a
positive principal eigenfunction $\phi^{\dagger}$.  Thus, for $\psi$ as in
the definition of (strictly) stably outermost,
$$
\lambda \la \phi^{\dagger}, \psi \ra = \la L_v^{\dagger} \phi^{\dagger}, \psi
\ra = \la \phi^{\dagger}, L_v \psi  \ra \geq 0 (>0),
$$
Since $\la \phi^{\dagger}, \psi \ra > 0$, the lemma follows. 
\end{proof}
The stability properties discussed above always refer to a MOTS $\Sc$ with
respect to a fixed direction $v^{\alpha}$ normal to $\Sc$ or a fixed
hypersurface to which $v^{\alpha}$ is tangent. This raises the question as to
how the stability properties depend on this direction. Using definition
(\ref{vectorv}) we can change $v^{\alpha}$ by adjusting the function $V$,
which can take any value, and study the dependence of the principal
eigenvalue $\lambda$ on this function.  This yields the following

\begin{proposition}
\label{varlam}
Let $\Sc$ be a MOTS and let $\lambda_v$, $\lambda_{v'}$ be the principal
eigenvalues of the stability operators in the directions $v^{\alpha}$ and $
v'^{\alpha}$ defined by (\ref{vectorv}), with some functions $V$ and $ V'$
respectively. Let $\phi'$ and $\phi^{\dagger}$ be principal eigenfunctions of
$L_{v'}$ and $L^{\dagger}_{v}$ respectively. Then
\begin{equation}
\label{equlam}
(\lambda_v - \lambda_{ v'}) \la \phi^{\dagger}, \ \phi' \ra =  \la
 \phi^{\dagger}, ( V' - V) W \phi' \ra.
\end{equation}
\end{proposition}
\begin{proof} 
From (\ref{LvLk}) it follows that,
\begin{equation}
\label{LvLvprime}
L_v  = L_{v'}  + ( V' - V) W
\end{equation}
Hence, $0 =  \la \phi^{\dagger}, \left( L_v  -L_{v'}  - (V' - V) W \right)
 \phi' \ra$ and the proposition follows using point \ref{point:princ-ii} 
of Lemma \ref{principal}.  
\end{proof} 
Applying some trivial estimates to (\ref{equlam}) we find
\begin{align}
\label{lam1}
\lambda_{ v'} + \inf_{\Sc}[( V' - V)W] & \le \lambda_v \le \lambda_{ v'} +
\sup_{\Sc}[( V' - V)W].
\end{align}
The inequality (\ref{lam1}) implies in particular that, in spacetimes
satisfying the null energy condition, a  MOTS which is stable or strictly
stable with respect to $v^{\alpha}$ will have the same property  with respect
to all directions ``tilted away'' from the null direction $l^{\alpha}$
defining the MOTS, and it puts a limit to the allowed amount of  tilting of
$v^{\alpha}$  ``towards'' $l^{\alpha}$ which preserves stability or strict
stability.   In particular, if a MOTS $\Sc$ is (strictly) stable with respect
to some spacelike  direction, then it is also (strictly) stable with respect
to the null direction $- k^{\alpha}$  complementary to $l^{\alpha}$ in the
normal basis $\{l^{\alpha},k^{\alpha} \}$.

\section{The graph representation of MOTS}
\label{graph}

We now assume that $({\cal M},g)$ is foliated by  hypersurfaces $\Sigma_t$
defined as the level sets of a smooth function $t$.  Assume also  that one of
the  elements of the foliation, say $\Sigma_0$, contains a MOTS $\Sc_0$.  We
further assume that $\Sigma_0$ is transverse to the null normal vector
$l^{\alpha}$ on $\Sc$, i.e. $T_p M = T_p \Sigma_0 \oplus ( l^{\alpha} |_p
)$ for all $p \in \Sc_0$.  However, for most results we do not assume any
specific causal character for the foliation $\Sigma_{t}$ and we allow leaves
which are  spacelike, timelike, null or which change their causal character.
The transversality above allows us to fix $l^{\alpha}$ so that $l^{\mu}
\partial_{\mu} (t) = 1$.

The main goal of this paper is to examine under which conditions there is a
MOTT $\TT$ such that  $\Sc_t = \TT \cap \Sigma_{t}$ is a MOTS for all $t$.
For this purpose it is useful to define an abstract copy  $\widehat\Sc$ of
$\Sc_0$ detached from spacetime. We are looking for a one-parameter family of
smooth immersions $\chi_{t} : \widehat\Sc \longrightarrow \Sigma_{t}$  for
$t$ in some open interval  $I \ni 0$ such that $\Sc_t \equiv \chi_{t}
(\widehat\Sc)$ is a MOTS and the map
\begin{align*}  
\Psi : I \times \widehat\Sc & \longrightarrow
 {\cal M} \\  
(t,p) & \longrightarrow  i_{\Sigma_{t}} \circ \chi_{t} (p),
\end{align*} 
is smooth, where $i_{\Sigma_{t}}: \Sigma_t \hookrightarrow {\cal M}$ is the
natural inclusion.  In other words, the collection of MOTS should depend
smoothly on $t$.  Using arbitrary local coordinates $y^i$ on $\Sigma_t$,
$x^A$ in $\widehat\Sc$,  and writing $y^{\mu} = \{ t, y^i \}$, the immersion
$\Psi$ takes the local form
\begin{equation}
\label{emb}
\Psi^{\mu} (t, x^A ) = (t, \chi^i_t (x^A)).
\end{equation}
 For the null expansion we obtain
\begin{equation} 
\thl = l^{t}_{\mu} h^{AB}_t K^{\mu}_{AB} = h^{AB}_t l^t_{\mu} \left (
\Gamma^{(\Sc) \,C}_{t \, \, \,AB}  \frac{\partial \Psi^{\mu}}{\partial x^C}
- \frac{\partial^{2} \Psi^{\mu}}{\partial x^A \partial x^B} - \left
. \Gamma^{\mu}_{\nu\rho} \right |_{x=\Psi^{\alpha}} \frac{\partial
\Psi^{\nu}}{\partial x^A} \frac{\partial \Psi^{\rho}}{\partial x^B} \right ),
\label{thl}
\end{equation} 
where $h_{t \, AB}$ is the induced metric on $\Sc_t$, $\Gamma^{(\Sc) \,C}_{ t
\, \, \,AB}$ are the corresponding Christoffel symbols and
$\Gamma^{\mu}_{\nu\rho}$ are the spacetime Christoffel symbols. The vector
$l^t_{\alpha}(x^A)$ satisfies $l^{t=0}_{\alpha} = l_{\alpha}$ and solves
\begin{equation}
\label{null} 
l^{t}_{\mu} l^{t \, \mu} = 0,  \quad l^t_{\mu} \frac{\partial
\Psi^{\mu}}{\partial x^A} =0, \qquad l^{t \, \mu} \partial_{\mu} t  = 1.
\end{equation}
under the condition that $l^t_{\alpha}$ depends continuously (and hence
smoothly) on $t$.  The condition we want to solve is $\thl=0$ which is a
scalar condition on the immersion $\chi_t$.  Since the problem is
diffeomorphism invariant, we need to choose coordinates in order to convert
the  geometric problem into a PDE problem. The idea is to construct a
suitable coordinate system adapted to the initial MOTS $\Sc_0$ and to
restrict the class of allowed MOTS to be suitable graphs in this coordinate
system.  Since  $\Sc_0$ may have self-intersections, the coordinate system
and the graphs we will use are only local in the sense that for any point $p
\in \widehat{\Sc}$, there is a spacetime  neighbourhood ${\cal V}_p$ of its
image $\Phi(p)$ such that the spacetime metric takes a special form in
suitable coordinates adapted  to the surface (strictly speaking, to the
connected component of $\Sc_0 \cap {\cal V}_p$ containing
$\Phi(p))$. In the following lemma we introduce this adapted coordinate
system.  In order to avoid cumbersome notation, we will assume $\Sc_0$ to be
embedded for this lemma. Since the applications of this result below  will
always be local on $\Sc_0$ in the sense just described (and all immersions
are locally embeddings) this suffices for our purposes.

The following lemma does not require $\Sc_0$ to be a MOTS.
\begin{lemma} 
\label{coordinates}
Let $(\M,g)$, $\Sigma_t$ and $t$ be as before.  There exists a spacetime
neighbourhood ${\cal V}$ of a smooth, embedded closed spacelike 2-surface
$\Sc_0 \subset \Sigma_0 \subset {\cal M}$, local coordinates $\{t,r,x^A\}$ on
${\cal V}$ adapted to the foliation  $\{\Sigma_{t}\}$ and functions
$Z,\varphi,\eta^A$ and $h_{AB}$ such that the metric can be written as
\begin{equation}
\label{met}
g_{\mu\nu} dx^{\mu} dx^{\nu} = 2 e^Z dt dr + \varphi dr^2 + h_{AB} \left (
dx^A + \eta^A dr \right ) \left ( dx^B + \eta^B dr \right ),
\end{equation}
where $\Sc_0 \cap {\cal V} = \{t = 0, r=0 \}$, $Z(t=0,r=0,x^A)=0$ and
$h_{AB}$ is a positive definite $(n-2)$-matrix.  The function $\varphi$ may
have any sign (even a changing one) reflecting  the fact that the
hypersurfaces ${\Sigma_t}$, which coincide with the $t=const$ surfaces, are
of any causal character.
\end{lemma}
\begin{proof} 
Let $x^A$ be a local coordinate system on $\Sc_0$ and let
$v^{\alpha}$ be a smooth, nowhere zero vector field on $\Sc_0$ orthogonal to
this surface, tangent to $\Sigma_0$, and satisfying $v^{\alpha}l_{\alpha} =
1$.  We extend $v^{\alpha}$ to a smooth vector field (still denoted by $v$)
in a  neighbourhood of $\Sc_0$ within $\Sigma_0$ and define coordinates $(r,
x^A)$ on this  neighbourhood by solving the ODE (again with slight abuse of
notation)
$$
\left. \begin{array}{cc} v (r)=1, & r|_{\Sc_0} = 0,  \\ 
v ( x^A ) = 0, & x^A
       |_{\Sc_0} = x^A 
       \end{array} \right \}.
$$
In some tubular neighbourhood of $\Sc_0$ within $\Sigma_0$, this solution
defines a smooth  coordinate system in which we have $v^{\alpha}
\partial_{\alpha} =
\partial_r$.  Furthermore, for small enough $r_0$ the sets $\Sc_{r_0} = \{ r
= r_0 \}$  define spacelike, closed, codimension-two surfaces embedded in
$\Sigma_0$ and diffeomorphic to $\Sc_0$.  After restricting the range of $r$
if necessary, we can also extend $l^{\alpha}$ to a nowhere vanishing null
vector field on $\Sigma_0$ orthogonal to each level surface $\{\Sc_r \}$ and
satisfying $l^{\alpha} \partial_{\alpha} t =1$.

We finally consider the null geodesics with tangent vector $l^{\alpha}$
whose ``length'' is fixed by $l^{\mu} \partial_{\mu} t = 1$ everywhere in a
suitably small neighbourhood of $\Sc_0$ in $\M$. On this neighbourhood ${\cal
V}$ define functions $\{r, x^A \}$ by solving $l^{\alpha} (r)=0$,
$l^{\alpha}(x^A)=0$ and so that $r |_{\Sigma_0}$ and $x^A |_{\Sigma_0}$
coincide with the functions with the same name defined above.  The functions
$\{ t, r, x^A \}$ define a coordinate system on ${\cal V}$.  Since
$l^{\alpha} \partial_{\alpha} = \partial_t$
everywhere we  immediately have $g_{tt}=0$ for the
metric. On  $\Sigma_0 = \{t=0\}$ we have $(\partial_t, \partial_{x^A})=0$,
where  $(\,, \,)$ denotes scalar product with $g$. The geodesic equation
$l^{\alpha} \nabla_{\alpha} l^{\beta} = b l^{\beta}$ becomes $\partial_t
g_{\mu t} = b g_ {\mu t}$ in this coordinate system ($b$ need not be zero as
the null vector $l^{\alpha}$ has already been chosen). Hence $g_{t x^A} =0$
and $g_{tr} > 0$ on this neighbourhood (because $g_{tr} = 1$ on $\Sc_0$).
\end{proof} 
In terms of the coordinates (\ref{met}), we can consider local graphs $\Sc_t$
given by $r = f(x^A)$ on the surface $\{r = 0\} \subset \Sigma_t$. 
Restricting the allowed immersions (\ref{emb}) to those local
graphs,   (\ref{thl}) becomes an operator on $f$, which we call $\tht[f]$.
We formulate the ellipticity property of this operator in the subsequent
lemma.  We use the shorthands $f_A = \partial_{x^A} f$ and $\rho =  [\varphi
f^A f_A + \left ( 1 + \eta^A f_A \right )^2 ]_{r=f}$, where $\varphi$ and
$\eta$ are the metric coefficients in (\ref{met}) and capital Latin indices
are moved with $h_{AB}$ and its inverse $h^{AB}$.

\begin{lemma}
\label{unif}
Consider codimension-two immersions defined locally by a smooth function
$f(x^A)$ according to
\begin{align*}
\chi_{t}[f] :    \widehat\Sc & \longrightarrow \Sc_t \subset {\cal M}\\ 
           x^A &\longrightarrow \left (t=t, r = f(x^A), x^A = x ^A \right).
\end{align*}

If  $f^C f_C < C$  for a suitable positive constant  $C$,
the operator $\tht[f]$ is
quasilinear and uniformly elliptic.
\end{lemma}
\begin{proof} 
We give explicitly some steps of this straightforward
calculation, which requires some care if $\Sigma_t$ is not spacelike.  We
define the quantities
\begin{equation}
\gamma   =  \varphi^{-1} e^{2Z} \left. \left( \rho^{-\frac{1}{2}}(1 + \eta^A
f_A) - 1 \right) \right|_{r=f}, \qquad \nu  =  \rho^{-\frac{1}{2}} e^{Z}|_{r=f},
\end{equation}
which are smooth even at points where $\varphi$ vanishes  (i.e. where
$\Sigma_t$ becomes null).  The one-form $l^t_{\alpha}$ satisfying
(\ref{null}) is
$$
l^t_{\alpha} \, dx^\alpha = \gamma dt + \nu \left ( dr - f_A dx^A \right),
$$
For the metric induced on the $\Sc_t$ and its inverse we have
\begin{align*}
h_{t\, AB} & =  \left. \left(h_{AB} + \varphi f_A f_B + f_A \eta_{B} + f_B
\eta_A + \eta^{C}\eta_{C} f_A f_B \right) \right|_{r=f}, \\ 
h_{t}^{AB} & = 
h^{AB} + \rho^{-1} \left.  \left[ f^D f_D \eta^A \eta^B - \varphi f^A f^B -
\left ( 1 + \eta^D f_D \right ) \left ( \eta^A f^B + f^A \eta^B \right)
\right] \right|_{r=f}.
\end{align*}
Note that $h_t^{AB}$ is positive definite wherever $\rho > 0$.  The operator
$\tht[f]$ can now be written explicitly
\begin{equation}
\tht [f] = - \nu h_t^{AB} f_{AB}  - \left. \left( l^t_{\mu}
(\Gamma^{\mu}_{rr} h_t^{AB})  f_A f_B  + 2 l^{t}_{\mu} \Gamma^{\mu}_{rA}
h^{AB}   f_B + l^{t}_{\mu} \Gamma^{\mu}_{AB} h^{AB} \right) \right|_{r=f} \,,
\label{ql} 
\end{equation}
where $f_{AB} = \partial_{x^A} \partial_{x^B} f$. Choosing $C$  small enough
it follows that $\rho > \epsilon$ for a positive $\epsilon$.  The assertion
of the lemma is verified easily by estimating the eigenvalues of $\nu
h^{AB}_t$.  
\end{proof} 
\begin{rmk}
We note that for quasilinear elliptic equations such as
$\tht[f] = 0$ there hold regularity results. In particular, if we required
$C^{2,\alpha}$ for the function $f$  instead of smooth as in the definitions
above,  the latter differentiability would in fact follow (see e.g. Sect. 8.3
of Evans \cite{LE}).
\end{rmk} 
We also remark that the stability operator (\ref{Lv}) coincides with the
linearization  of the quasilinear operator (\ref{ql}) and can be  obtained
from the latter by making $f$  infinitesimal. However, the expression
(\ref{ql}) requires a choice of coordinates. On the other hand (\ref{Lv}) was
derived in a covariant manner and therefore holds independently of the
coordinates on $\Sc$ or on the spacetime.

\section{Barrier properties of stable MOTS}
\label{barrier}

Let $\TT$ be a marginally outer trapped tube, i.e. a 3-surface foliated by
marginally outer trapped surfaces. 
Then there exists a positive variation along $\TT$ such that
$\delta_{\psi v} \thl = 0$. According to the definition in  section
\ref{stability} the MOTS foliating a MOTT are stable with respect to
the MOTT, but they are not strictly stable as  $\delta_{\psi v} \thl = L_v
\psi = 0$  implies that the principal eigenvalue vanishes. In general
situations,  since the variation of $\thl$ is essentially a derivative of the
$\thl$'s of adjacent surfaces, we expect that embedded strictly stable MOTS
are ``barriers'',  i.e. local boundaries separating regions containing weakly
outer trapped and weakly outer untrapped surfaces. 
The difficulty in showing this lies
again in the fact that strictly stable MOTS had to be defined in terms  of a
sign condition on a {\it single} variation, while they should be barriers for
{\it all} weakly outer  trapped surfaces in a neighbourhood.  

Below we give a result on this issue, 
which in fact requires the full machinery developed in
the preceding sections, namely the properties of the principal eigenvalue and
eigenfunction of the  stability operator as well as the maximum principle
applied to the quasilinear elliptic equation (\ref{ql}) representing the MOTS
as a graph. The maximum principles for non-linear operators normally  require
{\it uniform ellipticity} (c.f. eg. \cite[Sect. 10]{GT}) which is {\it not}
the case for operators of prescribed mean  curvature such as (\ref{ql})
when $f^C f_C$ is not small (c.f. Lemma \ref{unif}). For
this reason we prove the following theorem ``from scratch''. 
To state the
result, we recall from \cite{AMS} the definition of ``locally outermost''.
Note that in the definition and in the theorem below we require that the MOTS
is  embedded rather than immersed as in the previous sections.

\begin{definition} 
\label{locout}
Let $\Sc$ be a MOTS embedded in a hypersurface $\Sigma$.  $\Sc$ is called
{\bf locally outermost} in $\Sigma$, iff there exists a two-sided
neighbourhood of $\Sc$ in $\Sigma$ such that its exterior part does not
contain any weakly outer trapped surface.
\end{definition} 

\begin{theorem}\label{barr}
\begin{enumerate}
\item \label{point:ss}
An embedded, strictly stably outermost surface $\Sc$  is locally
outermost.  Moreover, $\Sc$ has a two-sided neighbourhood $\U$ such that no
weakly outer trapped surface  contained in $\U$ enters the exterior of $\Sc$
and no weakly outer untrapped surface contained  in $\U$ enters the interior
of $\Sc$.
\item \label{point:ss2}
A locally outermost surface $\Sc$ is stably outermost.
\end{enumerate}
\end{theorem}
\begin{proof} 
The first statement of \ref{point:ss} is in fact contained in the second
one.  To show the latter, let $\phi$ be the positive principal eigenfunction
of $L_v$.  Since $L_v \phi > 0$ by assumption, flowing $\Sc$ in $\Sigma$
along any extension of $\phi v^{\alpha}$ within $\Sigma$ produces a family
$\Sc_\sigma$, $\sigma \in (-\eps,\eps)$ for some $\eps > 0$. By choosing
$\eps$ small enough, the $\Sc_\sigma$ have $\thl |_{\Sc_\sigma} < 0$  for
$\sigma \in (-\eps, 0)$ and $\thl |_{\Sc_\sigma} > 0$ for  $\sigma \in
(0,\eps)$. We can now take $\U$ to be the neighbourhood of $\Sc$ given by
$\U = \cup_{\sigma \in (-\eps,\eps)} \Sc_\sigma$.  We first express the
expansion of the level sets of $\sigma$ in terms of connection
coefficients. Setting $f = \sigma = const. > 0$ in (\ref{ql}), only the last
term survives and we have
\begin{equation}
\label{thref}
0 <  \tht [\sigma] = - \left. l^t_{\mu} \Gamma^{\mu}_{AB}
h^{AB}\right|_{f = \sigma}
\end{equation}
and analogously for $\sigma < 0$.

Now let $\BB$ be a weakly outer trapped surface ( i.e. $\thl|_{\BB} \le 0$)
contained in $\U$ which enters the exterior part of $\U$. Let $p$
be the point where $\sigma |_{\BB}$ achieves
a maximum value. In a small neighbourhood of $p$, $\BB$ is 
a graph given by a function $f$ which achieves its maximum at $p$.
From (\ref{ql}),
 \begin{equation}
\label{thf}
0 \ge  \thl |_{\BB} = \tht[f] =
- \nu h_t^{AB} f_{AB} - \left. \left( l^t_{\mu}
(\Gamma^{\mu}_{rr} h_t^{AB})  f_A f_B + 2 l^{t}_{\mu} \Gamma^{\mu}_{rA}
h^{AB}   f_B + l^{t}_{\mu} \Gamma^{\mu}_{AB} h^{AB} \right) \right|_{r=f}.
\end{equation}
At $p$, $f_A|_p = 0$ and
 $h^{AB} f_{AB}|_p \le 0$. Since $\sigma = const$ and $\BB$ have a common
 tangent plane at $p$,  the last term in (\ref{thf}) coincides with
 (\ref{thref}), which yields the contradiction (we also use
$\nu |_p = e^Z |_p$)
\begin{equation}
  \tht [f]_p = - e^Z h^{AB} f_{AB}|_p -  \left. l^{t}_{\mu} \Gamma^{\mu}_{AB}
h^{AB} \right|_p > 0.
\end{equation}  
Hence $\BB$ cannot enter the exterior part of $\U$, and analogously  weakly
outer untrapped surfaces cannot enter the interior.

To show \ref{point:ss2}, 
assume $\Sc$ is locally outermost but not stably outermost.
From Proposition \ref{lem:eigen}, the principal  eigenvalue $\lambda$ is
then negative.  Arguing as above one constructs a foliation outside $\Sc$
with leaves which are outer trapped near $\Sc$, contradicting the
assumption. 
\end{proof} 

Returning to the example of a MOTT mentioned above, it follows that a MOTS
within a MOTT is not only not strictly stable 
with respect to the direction tangent to the MOTT,
as we already knew, but moreover not locally outermost.

The barrier arguments suggest that a region of a slice $\Sigma_t$ 
bounded by an outer trapped surface $\Sc_1$ and 
outer untrapped surface $\Sc_2$ contains at least one MOTS. 
This has been proven recently by Andersson and Metzger \cite{AM2}, based on
an argument proposed by Schoen \cite{RS}. 

\section{Symmetries}
\label{symmetries}

As examples one normally considers spacetimes with symmetries, i.e.
invariant under the action of some symmetry group.  If the symmetry is
spacelike, it is natural to consider spacelike foliations with the same
symmetry.  In order to locate the MOTS contained in the leaves and to compute
the principal eigenfunctions and eigenvalues, it is useful to know whether
the MOTS and the eigenfunctions inherit the symmetries from the ambient
geometry. In this section we address these questions, starting with the MOTS.

\begin{theorem} 
\label{iso}
Let $\Xi$ be a local isometry of $({\cal M},g)$  (i.e.  ${\cal
L}_{\xi}g_{\mu\nu} = 0$ for the Lie-derivative ${\cal L}_{\xi}$ w.r. to the
corresponding Killing field $\xi^{\alpha}$), and let  $\cal S$ be a MOTS
which is stable with respect 
to a normal direction $v^{\alpha}$ such that the normal
component $\xi^{\bot \alpha}$ of $\xi^{\alpha}$ satisfies  $\xi^{\bot \alpha}
= \psi v^{\alpha}$ for some function $\psi$.  Then either $\Xi$ leaves the
MOTS invariant (i.e. $\xi^{\alpha}$  is tangential to the MOTS), or
$\Xi(\Sc)$ is a MOTT.
\end{theorem}
\begin{proof} 
Clearly $\xi^{\alpha}$ leaves the MOTS invariant iff $\xi^{\bot
\alpha}$ and hence $\psi$ are identically zero. Assume this is not the case.
From (\ref{varv}), $0 = \delta_{\xi} \thl = \delta_{\psi v} \thl  =  L_v
\psi$ shows that $\psi$ is an eigenfunction of $L_v$ with eigenvalue zero.
Since the MOTS is stable with respect to $v^{\alpha}$,  $\psi$ is the unique (up to a
constant) principal eigenfunction of $L_v$ and has therefore a constant sign,
which after reversing $\xi^{\alpha}$ if necessary can be taken to be
positive. This implies that $\Xi$  in fact generates a MOTT as stated.
\end{proof} 

Note that this proof shows in particular that if $\Sc$ is strictly stable
then  $\Sc$ must remain invariant under the isometry.

If $\Sc$ lies in some hypersurface $\Sigma$ invariant under the isometry,
then Theorem \ref{iso} implies that either $\Xi$ leaves $\Sc$ invariant, or
$\Xi(\Sc) \subset \Sigma$. Clearly, in the second case $\cal S$ is not
strictly stable. Moreover, it is not locally outermost in the sense of
Definition \ref{locout}.

We also remark that the range of validity of Theorem \ref{iso} can be
extended with the help of Proposition \ref{varlam}. It suffices to require
that the MOTS is stable with respect to any direction normal to $\Sc$ which
lies ``between  $l^{\mu}$ and  $\xi^{\bot \mu}$''  (which is in fact a
conical segment).

The following theorem on the symmetry of the principal eigenfunction (``ground
state'') is well-known for self-adjoint operators, with numerous physical
applications.  We have formulated it here for the general linear elliptic
operator (\ref{ell}),  and therefore it holds in particular for the stability
operator (\ref{Lv}).  
\begin{theorem} 
\label{invariance}
If $L$ is invariant under a 1-parameter group of isometries 
generated by $\eta^{\alpha}$,
 the principal eigenfunction $\phi$ is invariant as well.
\end{theorem}
\begin{proof} 
The invariance of $L$ under the isometry implies that this
operator commutes with the action of the corresponding  Lie derivative ${\cal
L}_{\eta}$. Hence, if $\phi$ is an eigenfunction with principal eigenvalue
$\lambda$,  it follows that
\begin{equation}
L {\cal L}_{\eta} \phi = {\cal L}_{\eta} L \phi = \lambda {\cal L}_{\eta} \phi
\end{equation}
which means that ${\cal L}_{\eta} \phi$ is a principal eigenfunction as
well. But from Lemma \ref{principal}(i),  the latter is unique up to a
factor, i.e. ${\cal L}_{\eta} \phi = \alpha \phi$ for some real number
$\alpha$.  This expression integrates to zero on $\Sc$ because
${\cal L}_{\eta} \phi  = D_{A} ( \phi \eta^{A} )$.
Since $\phi$ has constant sign,
it follows that $\alpha=0$. This proves the assertion. 
\end{proof} 
As an example we consider axially symmetric data on $\Sigma$  with Killing
vector $\eta^{\alpha}=\partial/\partial \varphi$,   which has in addition a
$(t,\varphi)$ symmetry, i.e. invariance under simultaneous sign reversal of
$\varphi$ and $t$.  This is the case, in particular, for data on a $t=const.$
slice of the Kerr metric in Boyer-Lindquist coordinates.  It turns out that
the vector $s_A$ introduced in Sect. \ref{variation} is proportional to the axial Killing vector and
divergence-free, i.e.  $s_A = z_A$ in the notation of Sect. \ref{elliptic}.  By Theorem
\ref{invariance} the principal eigenfunction is axially symmetric and
therefore satisfies a second order ODE. Finding the principal eigenvalue and
eigenfunction is therefore equivalent to solving  a one-dimensional
Sturm-Liouville problem.  Moreover, the first order term in the stability
operator vanishes when acting on this eigenfunction.  However, this does not
imply that the stability operator is self-adjoint in this situation or that
the principal eigenvalue coincides with any of the symmetrized eigenvalues
$\lambda_s$ or $\lambda_z$.

\section{Existence and properties of MOTTs}
\label{mott}

The main objective of this paper is to show that MOTS ``propagate'' from a
given slice to adjacent leaves of a given foliation to form a MOTT.  Setting
$\thl = 0$ in (\ref{thl}) determines the location of a MOTT in terms of
immersions. These may be viewed as defining 
a family of quasilinear elliptic equations which do not
contain any derivatives along the presumptive ``evolution''
direction. Since we have assumed that $\Sc_0$ is marginally outer
trapped, the elliptic equation is satisfied for $t=0$ and we can adopt a
perturbational approach.  In section \ref{graph} we have derived the
``graph'' representation (\ref{ql}) in a special coordinate system,  
which will be used in the existence proof. The proof makes 
use of a general existence
result \cite{ADN} which shows  the existence of solutions  near a known
solution of a one-parameter-family of non-linear  differential equations (or
systems)  of arbitrary order and  arbitrary type whose linearization is
elliptic.  Since we will apply this result to a 1-parameter family of
quasilinear elliptic equations for which regularity results hold, we
will formulate the result here for smooth functions only. Equations of
general type require some care  regarding differentiability.

\begin{lemma}[\cite{ADN,AB}] \label{Nonlinear}
Let
\begin{equation}
\label{gen}
F(x, u, \partial u, \partial^2 u; t) = 0
\end{equation}
be a 1-parameter family of quasilinear, second-order differential equations,
where F is a smooth function of all arguments.  Assume that $u_0$ is a smooth
solution of (\ref{gen}) for $t=t_0$, and that the linearized equation around
$u_0$, $Lv = f$ is elliptic and has a unique solution for any smooth $f$.
Then (\ref{gen}) has a unique smooth solution for $t \in ( t_0 -\epsilon, t_0
+ \epsilon)$ for some $\epsilon > 0$.
\end{lemma}

This result is an easy application of the implicit function theorem;
alternatively, the latter theorem can also be applied directly to (\ref{ql}),
as we suggested in \cite{AMS}. The result is, in view of the above discussion,
that a  {\it smooth} horizon
exists in some  {\it open} neighbourhood of the given slice. More precisely, 
we have the following theorem and a corollary.
\begin{theorem} 
\label{ex1} 
Let $(\M, g_{\alpha\beta})$ be a spacetime foliated by smooth hypersurfaces
$\Sigma_t$, $t \in [0,T]$, and assume 
that $\Sigma_0$ contains a smooth, immersed,
strictly stable MOTS $\Sc_0$. Then for some $\tau \in (0,T]$ there is a smooth
adapted MOTT 
$$
\TT_{[0,\tau)} = \Phi(\Sc_0 \times [0,\tau))
$$
such that for each $t \in [0,\tau)$, $\Sc_t = \Phi(\Sc_0,t)$ is a smooth,
  immersed, strictly stable MOTS with $\Sc_t \subset \Sigma_t$. 
\end{theorem}

\begin{corollary}
\label{tang1}
The  MOTT $\TT_{[0,\tau)}$ constructed in Theorem \ref{ex1}  is nowhere tangent
to the leaves $\Sigma_t$ of the foliation.
\end{corollary}
\begin{proof}
Let $\{t,r,x^A\}$ be the local coordinate system introduced in
Lemma \ref{coordinates} and restrict the allowed MOTT to be local graphs on
${\cal S}_0$ as described in section \ref{graph}. Thus, we are looking for
functions $f(t,x^A)$ which satisfy $\tht[f] =0$ for all $t$ near $t=0$ and
satisfying $f(0,x^A)=0$.  In order to apply Lemma \ref{Nonlinear} we  only
need to check that the linearization of $\tht[f]$ is elliptic and invertible.
The immersion $\Phi$ defining the MOTT is given in this coordinate system by
$(t,r=f(t,x^A),x^A)$.  Recall that 
$l^{\alpha} \partial_{\alpha} = \partial_t$ 
and that  on
$\Sc_0$ the vector $v^{\alpha} \partial_{\alpha} = \partial_r$ satisfies $l^{\alpha}
v_{\alpha}=1$.  Linearizing the operator around $(t=0,f=0)$ corresponds to
fixing $t=0$ and making $f$ infinitesimal, or more precisely to evaluating 
$L (f') \equiv \partial_{\epsilon} \theta_{t=0}[f_{\epsilon}] |_{\epsilon=0}$
where $f_{\epsilon}$ depends smoothly on $\epsilon$ and satisfies
$f_{\epsilon} |_{\epsilon=0} =0$, 
$\partial_{\epsilon} f_{\epsilon} |_{\epsilon=0}=f'$.  It
follows that $L$ corresponds to  performing a geometric variation of $\thl$
along the vector $f' \partial_r$.  The general variation formula in
section \ref{variation} gives $L (f') = \delta_{f' v} = L_v (f')$.  Since by
assumption, $\Sc_0$ 
is strictly stably outermost with respect to  $v$, it
follows that $L_v$ and therefore $L$, is invertible. Existence of $f(t,x^A)$
for small $t$ solving $\tht[f]=0$ follows readily from Lemma
\ref{Nonlinear}. Moreover, the ``evolution'' vector to the MOTT, $q^{\alpha}
\equiv \Phi_{\star}(\partial_t) = \partial_t +  
f' \partial_r$ 
is by construction nowhere tangent to $\Sigma_t$, which proves
the corollary. 
\end{proof} 

As a complement to this Corollary we now discuss the "tangency" property of the 
MOTT when the principal eigenvalue $\lambda_{\tau}$ goes to zero, 
assuming smooth convergence of the $\Sc_t$ to a stable MOTS  $\Sc_\tau$. 

\begin{theorem}
\label{thm:lars-ex2} 
Let $(\M, g_{\alpha\beta})$ be a spacetime containing a smooth reference
foliation  
$\{\Sigma_t\}_{t\in[0,T]}$. 
Assume that $\Sigma_0$
contains a smooth, immersed, strictly stable MOTS $\Sc_0$.  Assume
furthermore that the adapted 
MOTT ${\TT}_{[0,\tau)}$  through $\Sc_0$ constructed in
Theorem \ref{ex1} is such that as $t \to \tau$, the surfaces $\Sc_t$ converge
to a smooth, compact, stable MOTS $\Sc_\tau$. 
Let $\lambda_t$ be the principal eigenvalue of the stability operator of
$\Sc_t$ and $\phi^{\dagger}_{t}$ the principal eigenfunction of its
adjoint. Assume that $\lambda^{-1}_t \langle \phi^{\dagger}_t,W |_{\Sc_{t}} \rangle$ 
has a limit (finite or infinite) as $t \to \tau$.

Then the closure 
$\TT_{[0,\tau]} = \TT_{[0,\tau)} \cup \Sc_\tau$ is an adapted MOTT.  
If in addition $\lambda_\tau = 0$ and 
$\la \phi^\dagger_{\tau},
W  |_{\Sc_\tau} \ra  \neq 0$,
then $\TT_{[0,\tau]}$ 
is tangent to
    $\Sigma_\tau$ everywhere on $\Sc_\tau$. 

\end{theorem}
\begin{corollary}
\label{tang}
If the null energy condition holds on $\Sc_{\tau}$, $\lambda_{\tau}=0$
and $W |_{\Sc_{\tau}} \neq 0$ somewhere, then $\TT_{[0,\tau]}$
is tangent to $\Sigma_\tau$ everywhere on $\Sc_\tau$. 
\end{corollary}

\begin{proof}
If $\lambda_\tau > 0$, then by uniqueness the MOTT through
$\Sc_\tau$ constructed using Theorem \ref{ex1} agrees with $\TT_{[0,\tau)}$
for $t \in [0,\tau)$, and hence the closure $\TT_{[0,\tau]}$ is an
adapted MOTT. It remains to consider the case where $\lambda_\tau = 0$. 
    
By construction, $\Sc_{t} = \Sigma_{t} \cap \TT_{[0,\tau)}$ is a MOTS for
$t \in [0,\tau)$. 
Let $v^\alpha$ be the normal to $\Sc_t$ tangent to $\Sigma_t$
(we drop the subindex $t$ for simplicity).
Consider the ``evolution'' vector $q^{\alpha}$
of the MOTS within $\TT_{[0,\tau)}$, given by $q^{\alpha} \partial_\alpha =
\Psi_{\star}(\partial_t)  = \partial_t + (\partial_t f)
\partial_r$. Recalling that $l^{\alpha} \partial_{\alpha} = \partial_t$ 
it follows that the
normal  component of $q^{\alpha}$ can be decomposed as
\begin{equation}
\label{qdec1}
q^{\bot \alpha}  =   l^{\alpha} + u v^{\alpha} 
\end{equation}
for some function $u$. The variation of $\thl$ along $q^{\alpha}$ must vanish
(this is precisely the condition that $q^{\alpha}$ is tangent to the MOTT).
Hence, $u$ must satisfy
 \begin{equation}
\label{LC}
\Der_q \thl = \Der_{l + uv} \thl =   - W + L_v u =  0,
\end{equation}
c.f. Lemma \ref{varth}. 

Consider the equation $L_v u = W$. Referring to Lemma \ref{lem:linalg}, 
we decompose $u = \PP u + \QQ u$ and write
this as $u = z \phi + w$ for some real $z$. We have  
$L_v u = \lambda z \phi + A w$, where $A = L_v \QQ$ as in Lemma
\ref{lem:linalg}. 
It follows that 
\begin{equation}\label{eq:zlam} 
 z  = \lambda^{-1} \frac{\la \phi^\dagger, W \ra }{\la \phi^\dagger, \phi \ra}
\end{equation} 
and 
$$
A w = \QQ W .
$$
Since  for $t \in [0,\tau]$,  $L_v$ is uniformly elliptic and its
coefficients are uniformly smooth, \cite[Theorem 3.1, p. 208]{Kato}
applies to show that the spectral gap is lower semi-continuous in $t$, i.e.
$\liminf c_0(t) \geq c_0(\tau)$.
At $\tau$, the spectral gap is
positive. Thus, it follows that
there exists $\epsilon >0$ and a constant $c_0 > 0$
such that the spectral gap estimate
(\ref{eq:Amubound}) holds  
for $t \in (\tau - \epsilon,\tau]$. 
Recall that by construction $\QQ w = w$.  
Since by assumption $\lambda > 0$ 
for $t < \tau$,
we may apply (\ref{eq:Aest}) with $\lambda = 0$, to conclude that 
$\QQ w$, and hence also $w$, is 
uniformly bounded in Sobolev spaces.

Suppose now that
$\la \phi^\dagger_{\tau}, W |_{\Sc_\tau} \ra \neq 0$
and $\lambda_\tau = 0$. Then it follows that $z
\to \infty$, as $t \to \tau$, and hence that $u$ diverges. We see further
that in this case, 
$$
\lim_{t \to \tau} z^{-1} u  = \phi_{\tau} > 0,
$$
where $\phi_{\tau}$ is the eigenfunction of $\lambda_{\tau}$.
Define a normalized evolution
vector $\hatq^\alpha = z^{-1} q^\alpha$. 
Then 
$$
\hatq^\alpha \partial_\alpha = z^{-1} \partial_t + z^{-1} u
\partial_r. 
$$
It follows from the above that as $t \nearrow \tau$, the vectors
$\hatq^\alpha$ converge smoothly to a limit, which is proportional to
$\partial_r$.  By construction $\hatq^\alpha$
is a smooth vectorfield tangent to the MOTT $\TT_{[0,\tau)}$, and we have now
  shown that $\hatq^\alpha$ extends smoothly to the closure 
$\TT_{[0,\tau]} = \TT_{[0,\tau)} \cup \Sc_\tau$. 
In order to demonstrate that the closure is indeed a MOTT,
    it remains to renormalize the time parameter. 

Thus, let $\phi_s: \Sc_0 \to \TT_{[0,\tau]}$ be defined by the flow of
  the vector field $\hatq^\alpha$, and let $I = [0,s_*]$ denote the parameter
  interval required for this flow to sweep out $\TT_{[0,\tau]}$. This
  construction defines a smooth monotone map $\sigma: I \to [0,\tau]$. We can
  now extend the immersion defining the MOTT $\TT_{[0,\tau)}$ to an
  immersion $\Phi: \Sc_0 \times I \to \M$ such that for $s \in I$, we
  have $\Sc_{\sigma(s)} = \Phi(\Sc_0,s) \subset \Sigma_{\sigma(s)}$ for $s
  \in I$. It is clear from the construction that this defines a MOTT 
$\TT_{[0,\tau]} = \TT_{[0,\tau)} \cup \Sc_\tau$. 

It remains to consider the case when 
$\la \phi^\dagger_{\tau}, W |_{\Sc_\tau} \ra = 0$.
In this case, by
applying the same argument as above, we find that $u$ may, or may not diverge
as $t \to \tau$, depending on the detailed behavior of $\lambda$ and $W$ as
$t \to \tau$. 
In particular, we see from (\ref{eq:zlam}) that if
$\lambda^{-1} \la \phi^\dagger, W \ra$
diverges as $t \to \tau$, 
then also $z$ diverges, and we are essentially in
the situation considered above. On the other hand, if 
$\lambda^{-1} \la \phi^\dagger, W \ra$ tends
to a bounded limit as $t \to \tau$, then we are in a situation which is
analogous to the case with $\lambda_\tau > 0$.  
In either case, the normalized evolution vector $\hatq^\alpha$ 
converges as $t \to \tau$, and an argument along the lines above shows that
$\TT_{[0,\tau]} = \TT_{[0,\tau)} \cup \Sc_\tau$ is an adapted MOTT. 
\end{proof} 

In \cite{AMS} we claimed that the evolution of the MOTS can be continued ``as long
as the MOTS remain  strictly stably outermost''.  We emphasize, however, that
in \cite{AMS} the MOTS were taken to be  {\it smooth and embedded by definition}. 
Moreover, we have tacitly assumed that the  MOTT does not ``run off to infinity''. 
Hence, our control of the evolution may end not only when $\lambda$ goes to zero  
but also when the MOTS develop self-intersections or when we lose  compactness or  smoothness.  
In the present setup, we do not worry about MOTS developing self-intersections as we have allowed 
them from the outset. However, to show in particular that existence continues up to and
{\it including} $\Sigma_{\tau}$ (which was assumed in Theorems \ref{ex1} and
\ref{thm:lars-ex2}) we have to exclude the
other pathologies.  

To avoid that the MOTT ``runs off to infinity'' we simply
require that it is contained in a compact subset of $\M$. 
Note that it may still happen that the area of the MOTS
$\Sc_t$ grows without bound as we  approach  $\Sc_{\tau}$. In view of the
curvature estimates of \cite{AM1}, this can happen in the four-dimensional
case only if the MOTS ``folds up'' sufficiently. 
Assuming a uniform area bound, we have the following
result, which follows from the work in \cite{AM1}. 

\begin{proposition} 
\label{prop:lars-ex2} 
Let $(\M, g_{\alpha\beta})$ be a spacetime of dimension 4, 
with a spacelike reference foliation 
$\{\Sigma_t\}_{t \in [0,T]}$. Assume that $\Sigma_0$
contains a smooth, immersed, strictly stable MOTS $\Sc_0$. Let
$\TT_{[0,\tau)}$ be the adapted MOTT through $\Sc_0$ constructed in Theorem
\ref{ex1}, and let $\Sc_t = \TT_{[0,\tau)} \cap \Sigma_t$ be the leaf of
$\TT_{[0,\tau)}$ in $\Sigma_t$. Assume that for $t \in[0, \tau)$, the
leaves $\Sc_t$ have uniformly bounded area, and are contained in a
compact subset of $\M$. Then the $\Sc_t$ converge as $t \to \tau$ 
to a smooth compact surface $\Sc_\tau$ which is a smooth, immersed, stable MOTS.
\end{proposition}

If the dominant energy condition is satisfied, the area of the limit set of the MOTT in fact stays
bounded as long as the MOTS stay strictly stable, as a consequence of
the following lemma. (This computation is known in the context of the
topological results \cite{RN}). Recall that $v^{\alpha}$ is everywhere
linearly independent of $l^{\alpha}$.

\begin{lemma}
\label{intL}
In a 4-dimensional spacetime $(\M, g_{\alpha\beta})$ in which the dominant
energy condition holds, let $\Sc$ be a MOTS which is strictly stable with
respect to a spacelike or null direction $v$, 
with principal eigenvalue $\lambda$ and area
$|\Sc|$.  Then $\Sc$ is topologically a sphere and $\lambda |\Sc| \le 4 \pi$.
Moreover, if $\lambda |\Sc| = 4 \pi$ then $\Sc$ has constant curvature,
i.e. $R_{AB} = \lambda h_{AB}$.
\end{lemma}
\begin{proof} 
We take $\phi$ to be the principal eigenfunction of $L_v$,  and
we call  $y_A = -\phi^{-1} D_A \phi + s^{A}$ and $y^2 = y_A y^A$.  Taking
$\psi$ as the eigenfunction $\phi$ in (\ref{Lv}), we obtain
\begin{equation}
\label{lam}
\lambda = D_A y^A - |y|^2 + \frac{1}{2} R_{\Sc} - Y,
\end{equation} 
where $Y$ has been defined in (\ref{eqY}).  Integrating (\ref{lam}) and using
the Gauss-Bonnet theorem gives
\begin{equation}
\label{lamA}
\lambda |\Sc| = 2 \pi \chi- \int_\Sc (y^2 + Y),
\end{equation}
where $\chi$ is the Euler number.  The first assertion of the lemma now
follows since $\lambda > 0$ and $Y \ge 0$.  If $\lambda |\Sc| = 4 \pi $,
(\ref{lamA}) implies  $y_A \equiv 0$ and $Y \equiv 0$. Putting this back into
(\ref{lam}), we have $R_{\Sc} = 2\lambda$  which implies the statement of the lemma
as $R_{AB} = \frac{1}{2}R h_{AB}$ in 2 dimensions. 
\end{proof} 
 
We remark that the same method shows that in the case $\lambda = 0$,  the
MOTS can be a torus, which necessarily must be flat ($R_{AB} = 0$).
 
With the help of Lemma \ref{intL}, we can now sharpen Proposition
\ref{prop:lars-ex2} to formulate an existence result as follows.

\begin{theorem} 
\label{ex3} 
Let $(\M, g_{\alpha\beta})$ be a spacetime of dimension 4, in which the dominant
energy condition holds and which is foliated by smooth spacelike
hypersurfaces $\Sigma_t$, $t\in[0,T]$. Assume that $\Sigma_0$ contains a
smooth, immersed, strictly stable MOTS $\Sc_0$.  Assume further that the
MOTT $\TT_{[0,\tau)}$ through $\Sc_0$ constructed in Theorem \ref{ex1} 
is contained in a compact subset of $\M$ and that either
\begin{itemize}
\item[(i)] $\liminf_{t \to \tau} \lambda_t > 0$, {\underline or}
\item[(ii)] $\lim_{t \to \tau} \lambda_t = 0$ and $\limsup_{t \to \tau}
|\Sc_{t} | < \infty$. 
\end{itemize}
 Then, there is a smooth, compact, strictly stable MOTS
$\Sc_\tau$ in $\Sigma_\tau$ such that $\TT_{[0,\tau]} = 
\TT_{[0,\tau)} \cup \Sc_\tau$ is an
adapted MOTT, which is the closure of $\TT_{[0,\tau)}$.  
\end{theorem}

We remark that, as the surface $\Sc_{\tau} \subset \Sigma_{\tau}$ itself
satisfies the requirements of the theorem,  the evolution in fact continues in
$[0,T]$  ``as long as the MOTS $\Sc_{t} \subset \Sigma_{t}$ stay strictly
stable''.

We now reproduce and slightly extend the result \cite{AMS}  on the causal
structure of MOTT foliated by strictly stably outermost MOTS, which must be
either null {\it everywhere} or spacelike {\it everywhere}. 

\begin{theorem}
\label{achronal}
Let $(\M, g_{\alpha\beta})$ be a spacetime in which the null energy condition
holds,  which is foliated by smooth, locally achronal (i.e. spacelike or null
at each point) hypersurfaces  $\Sigma_t$ nowhere tangent to $l^{\mu}$, and
that $\Sigma_0$ contains a strictly stable MOTS $\Sc_0$. Then, the following
holds for any MOTT $\TT$
through $\Sc_0$.
\begin{enumerate}
\item 
The MOTT
$\TT$ is achronal in a neighbourhood of $\Sc_0$. 
\item
If $W$ does
not vanish identically, $\TT$ is spacelike everywhere near $\Sc_0$.
\item 
If $W$ vanishes identically on $\Sc_0$, $\TT$ is null everywhere
on $\Sc_0$.
\end{enumerate}
\end{theorem}
\begin{proof} 
Recall equation  (\ref{LC}) for the  normal variation vector
$q^{\bot \alpha} = l^{\alpha} + u v^{\alpha}$.   Applying points
\ref{point:MPii}  and
\ref{point:MPiii} of the maximum principle, Lemma \ref{MP},  
proves the results. 
\end{proof}

We finally comment on a possible alternative construction for MOTTs.
Consider a weakly outer trapped surface $\Sc$ on some initial slice and take the null cone
emanating from it.  By the Raychaudhuri equation (\ref{Ray}) if $W > 0$, the
null cone  cuts each subsequent slice on an outer trapped surface which gives
a trapped "barrier" $\Sc_1$. If we also assume an outer untrapped "barrier"
$\Sc_2$ outside $\Sc_1$, 
which in particular always exists near spacelike or null
infinity in asymptotically flat spacetimes, there is a
MOTS in the region
bounded by $\Sc_1$ and $\Sc_2$ by the results of \cite{RS,AM2}, 
c.f. the discussion in section \ref{barrier}. In this way, one can  show the
existence of an outermost MOTS
on subsequent slices. In general however, this collection of MOTS need
not be a MOTT because it may jump, as it always tracks the outermost MOTS on
each hypersurface $\Sigma_t$.

\subsection*{Acknowledgments}
We thank Jan Metzger for several helpful discussions. 
We also wish to thank Alberto Carrasco Ferreira for comments on the manuscript
and Helmuth Urbantke for improving Thm. \ref{invariance}.
We gratefully acknowledge the
hospitality and support received from the Isaac Newton Institute, Cambridge,
where part of the work was performed.
MM and WS were supported by the projects
FIS2006-05319 of the Spanish Ministerio de Educaci\'on y Tecnolog\'{\i}a
and SA010CO of the Junta de Castilla y Le\'on. MM also acknowledges
financial support from the Junta de Andaluc\'{\i}a under
project P06-FQM-01951. LA was partly supported by the NSF under contract
no. DMS 0407732, with the University of Miami.

\appendix

\section{Proof of Lemma \ref{gv}}\label{sec:defproof} 

Due to the linearity of the variation and the fact that  $q^{\, \|}$
generates a diffeomorphism of $\Sc$, which implies $\delta_{q^{\, \|}} \thl =
q^{\, \|} (\thl)$, we can assume $q  = q^{\bot}$ without loss of
generality. Since the variation is a local calculation we may assume that all
$\Sc_{\ts}$ 
are embedded. Let $k_{\ts}$ be a null, future directed normal
vector to $\Sc_{\ts}$ satisfying $k^{\alpha}_{\ts} l^{\beta}_{\ts}
g_{\alpha\beta}= -2$. We wish to extend $l^{\alpha}_{\ts}$ and
$k^{\alpha}_{\ts}$, for each value of $\ts$ to a one parameter family of
vector fields defined on a neighbourhood of $\Sc$. Let $\Omega_{\ts}$ be the
null hypersurface generated by $k^{\alpha}_{\ts}$ by affinely parametrized
geodesics. Extend first  $l^{\alpha}_{\ts}$ to $\Omega_{\ts}$ by parallel
transport along $k^{\alpha}_{\ts}$. Then extend $l^{\alpha}_{\ts}$ away from
$\Omega_{\ts}$ by affinely parametrized null geodesics and finally extend
$k^{\alpha}_{\ts}$ by parallel transport along $l^{\alpha}_{\ts}$.  Notice
that, in general, the two planes orthogonal to $\{{l}^{\alpha}_{\ts},
{k}^{\alpha}_{\ts}\}$ at each point do not define two-surfaces. On
$\Sc_{\ts}$ however they obviously do.  Being  $l^{\alpha}_{\ts}$ a geodesic field, the
null expansion $\thl_{\ts}$ can be rewritten as $\thl_{\ts}(p) =
(\nabla_{\alpha} l^{\alpha}_{\ts}) |_{\Phi(p,\ts)}$ for any $p \in \Sc$,
where $\Phi(p,\ts)$ is the variation of Sect. \ref{variation}.
Defining ${U^{\alpha}} = \partial_{\ts} {l}^{\alpha}_{\ts} |_{\ts=0}$ (with
partial derivative taken at a fixed spacetime point), we obtain directly from
the definition of the variation
\begin{equation} 
\delta_{{q}} \, \thl = \nabla_{\alpha} U^{\alpha} + q^{\beta} \nabla_{\beta}
\nabla_{\alpha} l^{\alpha} |_{\Sc} = \nabla_{\alpha} U^{\alpha} + q^{\beta}
\nabla_{\alpha} \nabla_{\beta} l^{\alpha} - \B G_{\alpha\beta} l^{\alpha}
l^{\beta}  + \frac{\C}{2} R_{\alpha\beta} k^{\alpha} l^{\beta} |_{\Sc}
 \label{vartheta1}
\end{equation} 
where $l^{\alpha} = l^{\alpha}_{\ts=0}$ and $k^{\alpha} = k^{\alpha}_{\ts=0}$
and the Ricci identity has been used in the second equality.  Let us next
determine the divergence of $U^{\alpha}$. From  that fact that
$l_{\ts}^{\alpha}$ is null for all $\ts$, it follows $U^{\alpha}
l_{\alpha}=0$ and hence  $U^{\alpha} = \A l^{\alpha} + U^A e^{\alpha}_A$ for
$\A$ as in the statement of the lemma and for suitable functions $U^A$.  Here
$e^{\alpha}_A$ is a basis of the orthogonal subspace to $l$ and $k$ at each
point.  Using the projector $h_{\alpha\beta} =  g_{\alpha\beta} + \frac{1}{2}
( l_{\alpha} k_{\beta} + k_{\alpha} l_{\beta} )$ it is easily found that the
divergence of any vector of the form $F^A e_{A}^{\alpha}$ is
\begin{equation}
\nabla_{\alpha} \left ( F^A e^{\alpha}_{A} \right ) |_{\Sc} = D_A F^A
|_{\Sc}  \label{diver}
\end{equation}
where we used the fact that $l^{\alpha} \nabla_{\alpha} k^{\beta} = 0$ and
$k^{\alpha} \nabla_{\alpha} l^{\beta} |_{\Sc} = 0$.  We need to determine
$U^A$ on $\Sc$: in local coordinates $y^{\alpha} (x^A, \ts)$ for the map
$\Phi$, orthogonality of $l^{\alpha}_{\ts}$ to $\Sc_{\ts}$ means
$$
g_{\alpha\beta} (y^{\mu} (\ts,x^B) ) \frac{\partial y^{\alpha}}{\partial x^A}
l^{\beta}_{\ts} (y^{\mu} (\ts, x^B))=0,
$$
Its $\ts$-partial derivative at $\ts=0$  gives, $(\nabla_{e_A} q, l ) + U_A
|_{\Sc}= 0$, i.e.  $U_A |_{\Sc} = - D_A \C + \C s_A$. From  $l^{\alpha}_{\ts}
\nabla_{\alpha} l^{\beta}_{\ts}=0$ it follows $U^{\alpha} \nabla_{\alpha}
l^{\beta} + l^{\alpha} \nabla_{\alpha} U^{\beta} =0$, which after
multiplication with $k_{\beta}$ and the fact  $k^{\alpha}$ is parallel along
$l^{\alpha}$ gives $l^{\alpha} \partial_{\alpha} \A = - U^A s_A$.  Thus,
using (\ref{diver})
\begin{equation} 
\nabla_{\alpha} U^{\alpha} |_{\Sc} = - \left . \Delta_{\Sc} \C + 2 s^A D_A \C +
\C \left (D_A s^A - s^A s_A \right ) + \A \thl \, \right |_{\Sc}.
\label{first}
\end{equation} 
We next consider the second term in (\ref{vartheta1}),  i.e.  $\B l^{\beta}
\nabla_{\alpha} \nabla_{\beta} l^{\alpha} - \frac{\C}{2} k^{\beta}
\nabla_{\alpha} \nabla_{\beta} l^{\alpha}$.  An integration by parts and
using that $l^{\alpha}$ is geodesic implies
$$
l^{\beta} \nabla_{\alpha} \nabla_{\beta} l^{\alpha} |_{\Sc} = -
\nabla_{\alpha} l_{\beta} \nabla^{\beta} l^{\alpha} |_{\Sc} = - K^{\mu}_{AB}
K^{\nu \, AB} l_{\mu} l_{\nu} |_{\Sc}.
$$
Decomposing $k^{\beta} \nabla_{\beta} l^{\alpha} = Q l^{\alpha} + r^A
e^{\alpha}_A$ and using the fact that $Q |_{\Sc} = r^A |_{\Sc} = 0$, another
integration by parts gives $ k^{\beta} \nabla_{\alpha} \nabla_{\beta}
l^{\alpha} |_{\Sc}= l^{\mu} \partial_{\mu} Q  - K^{\mu}_{AB} K^{\nu \, AB}
l_{\mu} k_{\nu} |_{\Sc}$.  In order to determine $l^{\alpha}
\partial_{\alpha} Q $, we  first note that $-2 Q = k^{\alpha} k^{\beta}
\nabla_{\alpha} l_{\beta}$.  Taking covariant derivative along $l^{\alpha}$
and using the Ricci identity we find $ 2 l^{\mu} \partial_{\mu} Q  |_{\Sc} =
l^{\alpha} k^{\mu} l^{\beta} k^{\nu} R_{\alpha\mu\beta\nu} |_{\Sc}$.
Collecting terms we get
\begin{equation} 
q^{\alpha} \nabla_{\alpha} \nabla_{\beta} l^{\alpha} |_{\Sc} = \left . - \B
K^{\mu}_{AB} K^{\nu\, AB} l_{\mu} l_{\nu}  + \frac{\C}{2} \left (
K^{\mu}_{AB} K^{\nu\, AB} l_{\mu} k_{\nu} - \frac{1}{2} l^{\alpha} k^{\mu}
l^{\beta} k^{\nu} R_{\alpha\mu\beta\nu} \right )   \right |_{\Sc}.
\label{second}
\end{equation}
For the last term in (\ref{vartheta1}), the definition of the projector
$h^{\alpha\beta}$ gives $R + 2 R_{\alpha\beta} l^{\alpha} k^{\beta} =
h^{\alpha\beta} h^{\mu \nu} R_{\alpha\mu\beta\nu} + \half l^{\alpha} k^{\mu}
l^{\beta} k^{\nu} R_{\alpha\mu\beta\nu}$. Making use of the Gauss identity
$h^{\alpha\beta} h^{\mu\nu} R_{\alpha\mu\beta\nu} |_{\Sc}= R_{\Sc} - {H}^2 +
K_{\mu \,  AB} K^{\mu \, AB} $ we get
\begin{equation}
\left . R_{\alpha\beta} l^{\alpha} k^{\beta} |_{\Sc} =   - G_{\alpha \beta}
l^{\alpha} k^{\beta} +  R_{\Sc} - {H}^2 -  K^{\mu}_{AB} K^{\nu \, AB} l_{\mu}
k_{\nu}  + \frac{1}{2} l^{\alpha} k^{\mu} l^{\beta} k^{\nu}
R_{\alpha\mu\beta\nu}  \right |_{\Sc}. \label{third}
\end{equation} 
Inserting (\ref{first}), (\ref{second}) and (\ref{third}) into
(\ref{vartheta1}) completes the proof. $\hfill \Box$

\section{Proof of Lemma 4.1 (sketch)}\label{sec:appendB}  
The Krein-Rutman theorem states that  on a Banach space $B$, a compact linear
operator $E$ that maps any non-zero element of a closed cone $K$ (i.e. a
topologically closed subset of $B$ closed under addition and multiplication
by non-negative scalars) into its topological  interior, necessarily has a
unique eigenvector $u$ in the interior of $K$ of unit norm.  Moreover, the
corresponding eigenvalue $\alpha$ is real and positive and any other element
$\beta$ of the spectrum of $E$ (complex in general) satisfies $\alpha \geq
|\beta|$ where $|\cdot |$ denotes the complex norm. A proof of this theorem
can be found in \cite{KR}.

In the case of the elliptic operator $L$ (\ref{ell}), let $\delta$
be a constant satisfying $\delta > \sup_{\Sc} c$ and define the operator $L' = L + \delta$.
The zero
order term is therefore positive everywhere and the PDE $L'f = g$ admits a
unique solution in $C^{2,\alpha}(\Sc)$ for any $g \in C^{0,\alpha}(\Sc)
\equiv B$. Let $Q : B \rightarrow B$ be defined by $Q(g) = f$ and let $K$ be
the set of non-negative functions (which is obviously a cone). The maximum
principle implies that if $g \in K$ and non-identically zero, then $f = Q(g)$
is strictly positive everywhere.  Thus, all the conditions of the
Krein-Rutman theorem are fulfilled and there exists a unique non-negative
function $\phi$ of unit $C^{0,\alpha}(\Sc)$ norm  satisfying $Q(\phi) = \alpha
\phi$. Since $\alpha$ is positive and $\phi$ is in the image of $Q$ it follows
that $\phi$ is strictly positive and in
$C^{2,\alpha}(\Sc)$. Elliptic regularity implies that $\phi$ is in fact
smooth. It follows that  $L \phi =  (\alpha^{-1} - \delta ) \phi$, so we have a
positive eigenfunction (unique up to rescaling) and a real eigenvalue
$\lambda = \alpha^{-1} - \delta$.

It only remains to show that any other eigenvalue of $L$ has larger or equal
real part. This does not follow directly from the Krein-Rutman theorem.
However, we use the following argument, which we adapt from Evans
\cite{LE}. Let $\psi$ be a (possibly complex) eigenfunction of $L$
with eigenvalue $\mu$. Define $u = \phi^{-1} \psi$. A
direct computation gives
\begin{eqnarray}
- D_A D^A u + 2 t^{\prime \, A} D_A u + (\lambda - \mu) u = 0,
\label{cc}
\end{eqnarray}
where $t^{\prime \, A} = t^A - D^A \phi$. Using also the complex conjugate of
(\ref{cc}) a short calculation gives
\begin{eqnarray*}
\left ( -D_A D^A + 2 t^{\prime A} D_A \right ) |u|^2 = 2 \left ( \Real (\mu)
 - \lambda \right ) |u|^2 - D_A u D^A \overline{u}  \leq 2 \left (
\Real (\mu) - \lambda \right ) |u|^2
\end{eqnarray*}
Thus, if $\Real(\mu) < \lambda $ the right-hand side is
non-positive and the maximum principle would imply $|u|^2=0$. Thus,  $\Real
(\mu) \geq \lambda$ as claimed.

Finally, the result on the adjoint is a trivial consequence of the positivity
of the principal  eigenfunctions, c.f. \cite{LE}. Explicitly, if
$\phi^{\dagger}$ and $\lambda^{\dagger}$ are the principal eigenfunction and
eigenvalue of $L^{\dagger}$, it follows
$$
0 = \langle L^{\dagger} \phi^{\dagger} , \phi \rangle - \langle
\phi^{\dagger}, L \phi \rangle =  \left ( \lambda^{\dagger} - \lambda \right
) \langle \phi^{\dagger}, \phi \rangle
$$
and $ \lambda^{\dagger} = \lambda $ as positive functions cannot be $L^2$
orthogonal.

\end{document}